\documentclass{article}

\usepackage[australian]{babel}
\usepackage{fullpage}

\usepackage{hyperref}

\usepackage{amssymb,amsmath,amsfonts}
\usepackage{amsthm}
\usepackage[all]{xy}
\usepackage[capitalise]{cleveref}

\usepackage{ogonek}

\allowdisplaybreaks
\numberwithin{equation}{section}

\usepackage{xcolor}
\definecolor{MyBlue}{cmyk}{1,0.13,0,0.63}
\definecolor{MyGreen}{cmyk}{0.91,0,0.88,0.52}
\newcommand{\mylinkcolor}{MyBlue}
\newcommand{\mycitecolor}{MyGreen}

\makeatletter
\def\@endtheorem{\endtrivlist}
\makeatother
\theoremstyle{plain}
\newtheorem{thm}{Theorem}[section]
\newtheorem{lem}[thm]{Lemma}
\newtheorem{prop}[thm]{Proposition}

\theoremstyle{definition}
\newtheorem{defn}[thm]{Definition}
\newtheorem{remark}[thm]{Remark}
\newtheorem{assumption}[thm]{Assumption}

\renewcommand{\eqref}[1]{\labelcref{#1}}
\crefname{thm}{Theorem}{Theorems}
\crefname{lem}{Lemma}{Lemmas}
\crefname{prop}{Proposition}{Propositions}
\crefname{coro}{Corollary}{Corollaries}
\crefname{defn}{Definition}{Definitions}
\crefname{remark}{Remark}{Remarks}

\hypersetup{%
  bookmarksnumbered=true,bookmarksopen=false,%
  plainpages=false,
  linktocpage=true,%
  colorlinks=true,breaklinks=true,%
  linkcolor=\mylinkcolor,citecolor=\mycitecolor,%
  pdfpagelayout=OneColumn,%
  pageanchor=true,%
}

\makeatletter
\def\thm@space@setup{%
  \thm@preskip=4pt plus 2pt minus 2pt
  \thm@postskip=\thm@preskip
}
\renewenvironment{proof}[1][\proofname]{\par
  \pushQED{\qed}%
  \normalfont \topsep4\p@\relax 
  \trivlist
  \item[\hskip\labelsep
        \itshape
    #1\@addpunct{.}]\ignorespaces
}{%
  \popQED\endtrivlist\@endpefalse
}
\makeatother

\usepackage{enumitem}
\setlist{topsep=4pt plus 2pt minus 2pt,partopsep=0pt,itemsep=2pt plus 2pt minus 2pt,parsep=0.5\parskip}
\setenumerate[1]{label=\arabic*)}

\newcommand{\MR}[1]{}

\let\OLDthebibliography\thebibliography
\renewcommand\thebibliography[1]{
  \addcontentsline{toc}{section}{\refname}
  \OLDthebibliography{BGM05}
  \setlength{\parskip}{0pt}
  \setlength{\itemsep}{0pt plus 0.3ex}
}

\RequirePackage{slashed}
\newcommand{\sD}{\slashed{D}}

\newcommand{\R}{\mathbb{R}}
\newcommand{\C}{\mathbb{C}}
\newcommand{\Z}{\mathbb{Z}}
\newcommand{\A}{\mathcal{A}}
\newcommand{\mH}{\mathcal{H}}
\newcommand{\mK}{\mathcal{K}}
\newcommand{\mJ}{\mathcal{J}}
\newcommand{\D}{\mathcal{D}}
\newcommand{\B}{\mathcal{B}}
\newcommand{\E}{\mathcal{E}}

\DeclareMathOperator{\tr}{tr}
\DeclareMathOperator{\Ad}{Ad}
\DeclareMathOperator{\Dom}{Dom}
\DeclareMathOperator{\SO}{SO}
\DeclareMathOperator{\GL}{GL}
\DeclareMathOperator{\Spin}{Spin}
\DeclareMathOperator{\Div}{div}
\DeclareMathOperator{\wres}{wres}

\renewcommand{\bar}[1]{\overline{#1}}
\newcommand{\Cliff}{{\mathrm{Cl}}}
\newcommand{\CCliff}{{\mathbb{C}\mathrm{l}}}

\newcommand{\dvol}{\textnormal{dvol}}
\newcommand{\pos}{\textnormal{pos}}
\newcommand{\til}[1]{\widetilde{#1}}
\newcommand{\la}{\langle}
\newcommand{\ra}{\rangle}
\newcommand{\into}{\hookrightarrow}
\newcommand{\contract}{\mathbin{\lrcorner}}
\newcommand{\bigmvert}{\,\big|\,}

\newcommand{\bS}{S}
\newcommand{\bP}{P}

\newcommand{\mattwo}[4]{
  \left(\!\!\!\begin{array}{c@{~}c}#1&#2\\ #3&#4\\\end{array}\!\!\!\right)
}

\title{Families of spectral triples and foliations of space(time)}
\author{
Koen van den Dungen%
\footnote{\emph{Present address:} Mathematisches Institut der Universit\"at Bonn, Endenicher Allee 60, D-53115 Bonn, \texttt{kdungen@uni-bonn.de}}\\[4mm]
{\normalsize SISSA (Scuola Internazionale Superiore di Studi Avanzati)}\\ 
{\normalsize Via Bonomea, 265, 34136 Trieste, Italy}
}
\date{}

\begin{document}

\maketitle

\begin{abstract}
\noindent
We study a noncommutative analogue of a spacetime foliated by spacelike hypersurfaces, in both Riemannian and Lorentzian signatures. 
First, in the classical commutative case, we show that the canonical Dirac operator on the total spacetime can be reconstructed from the family of Dirac operators on the hypersurfaces. 
Second, in the noncommutative case, the same construction continues to make sense for an abstract family of spectral triples. 
In the case of Riemannian signature, we prove that the construction yields in fact a spectral triple, which we call a product spectral triple. 
In the case of Lorentzian signature, we correspondingly obtain a `Lorentzian spectral triple', which can also be viewed as the `reverse Wick rotation' of a product spectral triple. 
This construction of `Lorentzian spectral triples' fits well into the Krein space approach to noncommutative Lorentzian geometry. 

\vspace{\baselineskip}

\end{abstract}

\section{Introduction}

Within the framework of Connes' noncommutative geometry \cite{Connes94}, the notion of spectral triples encompasses and generalises \emph{Riemannian} spin manifolds \cite{Con13}. 
Indeed, the canonical Dirac operator on a complete Riemannian spin manifold is (essentially) self-adjoint and elliptic, and therefore defines a spectral triple. 
One of the main open questions in noncommutative geometry (in particular, regarding its applications in physics) is how one should incorporate \emph{Lorentzian} manifolds into this framework. 
The canonical Dirac operator on a Lorentzian manifold is neither symmetric nor elliptic, and thus one needs to find a new framework in which to describe these operators. 
There are currently several tentative approaches to noncommutative Lorentzian geometry. 
One possible approach is based on the idea that there should be an abstract notion of `Wick rotation', which associates a (genuine) spectral triple to any `Lorentzian spectral triple' \cite{vdDPR13,vdDR16}. 
Another possible approach is to replace the Hilbert space by a Krein space \cite{Str06,Sui04,PS06,Bes16pre}, which is the natural way to describe spinors on Lorentzian manifolds. This Krein space approach provides a natural framework for the description of almost-commutative Lorentzian manifolds \cite{Bar07,vdD16}. 
Other lines of research into noncommutative Lorentzian geometry focus on the Lorentzian distance function \cite{Mor03,Fra10,Fra14,Fra18} (see also \cite{RW16,Min17pre}) or on the causality properties \cite{Bes09,FE13,FE14,FE15,BB17}. 
Finally, links between noncommutative geometry and quantum gravity have also been explored (see \cite{AG14} and references therein). 

In this article, we will consider a more constructive approach to the study of `Lorentzian spectral triples', modeled on a foliation of spacetime by spacelike hypersurfaces (see \cite{Haw97,Kop98,KP01,KP02,PV04} for earlier work in this direction). 
Given such a foliation, the Lorentzian Dirac operator $\sD$ can be decomposed into a family of (Riemannian) Dirac operators $\{\sD_t\}_{t\in\R}$ on the spacelike hypersurfaces, parametrised by the time-coordinate $t\in\R$. 
By analogy, one then describes a `Lorentzian spectral triple' using a family of spectral triples, parametrised by $\R$. 
Our main contribution to this approach, is that we provide a rigorous reconstruction of the total triple from the family of spectral triples. 

In fact, we will consider both Riemannian and Lorentzian signatures. 
First, in \cref{sec:Dirac_hyper}, we recall the description of the Dirac operator on a hypersurface, following \cite[\S3]{BGM05}. We focus our attention on even-dimensional space(time)s. 
In \cref{sec:prod_spacetime}, we introduce a class of \emph{product space(time)s} $Z = M\times\R$ which will be considered in this article. 
A product space(time) is only a \emph{topological} product of $M$ and $\R$; the geometry is allowed to be much more general than just a product geometry. 
We equip such a product space(time) with the canonical triple $\big( C_c^\infty(Z) , L^2(Z,\bS_Z) , \sD_Z \big)$ (which is a spectral triple if $Z$ is Riemannian). 
We can decompose this triple into a family of spectral triples $\big( C_c^\infty(M_t) , L^2(M_t,\bS_t) , \sD_{M_t} \big)$ on the hypersurfaces $M_t = M\times\{t\}$. 
Using parallel transport, we can view the Dirac operators on the hypersurfaces as a family of operators $\{\sD_t\}$ on a fixed Hilbert space $L^2(M_0,\bS_0)$. 
The reconstruction of the total Dirac operator $\sD_Z$ from the family $\{\sD_t\}_{t\in\R}$ requires an additional geometric object on the hypersurfaces: the lapse functions $N_t$. 
The Dirac operator $\sD_Z$ on $L^2(Z,\bS_Z)$ is then unitarily equivalent to the operator
\begin{align}
\label{eq:D_total}
\mattwo{0}{ i\tau_0 N_\bullet^{-\frac12} \partial_t N_\bullet^{-\frac12} + i \sD_\bullet}{i\tau_0 N_\bullet^{-\frac12} \partial_t N_\bullet^{-\frac12} - i \sD_\bullet}{0}
\end{align}
on $L^2\big(\R,L^2(M_0,\bS_0)\big)^{\oplus2}$, where $\tau_0=1$ if $Z$ is Riemannian, and $\tau_0=i$ if $Z$ is Lorentzian. 

In \cref{sec:prod_spec_trip}, we provide abstract axioms for families of spectral triples $(\A, \mH, \D_t)$ (generalising the Dirac operators on the hypersurfaces) and for families of \emph{lapse operators} $N_t$. 
We then define a new operator $\D$ by inserting these abstract families into \cref{eq:D_total}. 
This yields a new triple $(\A\odot C_c^\infty(\R), L^2(\R,\mH)^{\oplus2}, \D)$, which we interpret as describing the total \emph{noncommutative space(time)}. 
Our main result (\cref{thm:spec_prod}) is that, in the case of Riemannian signature ($\tau_0=1$), this new triple is in fact a spectral triple. We will refer to such triples as \emph{product spectral triples}. 
Moreover, we prove that such a product spectral triple represents the unbounded Kasparov product (over $C_0(\R)$) of the family $\big\{(\A, \mH, \D_t)\big\}_{t\in\R}$ (which yields an unbounded Kasparov $C_0(\R,A)$-$C_0(\R)$-module) with the standard spectral triple $\big( C_c^\infty(\R), L^2(\R),-i\partial_t \big)$ over the real line. 

Subsequently, in \cref{sec:Lor_prod} we define \emph{Lorentzian product triples} by applying the same formula \eqref{eq:D_total}, using instead the Lorentzian signature ($\tau_0=i$), to the abstract objects $(\A, \mH, \D_t)$ and $N_t$. 
The fact that \cref{eq:D_total} yields a spectral triple in Riemannian signature, justifies the idea that the Lorentzian version of \cref{eq:D_total} yields a Lorentzian analogue of spectral triples. 
Thus we have obtained an abstract construction for (noncommutative) `Lorentzian spectral triples'. 
These Lorentzian product triples can also be viewed as the `reverse Wick rotation' of (Riemannian) product spectral triples (using the terminology of \cite{vdDR16}). 
Moreover, we will show that our construction is compatible with the Krein space approach to noncommutative Lorentzian geometry. In particular, a Lorentzian product triple satisfies the abstract definition of a Lorentz-type spectral triple given in \cite{vdD16}. 

As a final remark, let us observe that in our construction of noncommutative `Lorentzian spectral triples', the time-coordinate always remains commutative. Thus, our construction does not capture noncommutative `Lorentzian spectral triples' in full generality. Instead, it should be viewed as a first step in this direction. 
We hope that our framework might lead to an `ADM-like' formalism, in which the time-evolution of noncommutative spaces can be studied.

\subsection*{Acknowledgements}

The author thanks Ludwik D\k{a}browski and Walter van Suijlekom for several interesting discussions related to this work. 
Thanks also to the referee for suggestions of improvement. 
This work is part of the project supported by the grant H2020-MSCA-RISE-2015-691246-QUANTUM DYNAMICS.

\section{The Dirac operator on a hypersurface}
\label{sec:Dirac_hyper}

\subsection{Clifford algebras}
\label{sec:Clifford_alg}

We shall start by briefly recalling the basics of finite-dimensional Clifford algebras. For more details, we refer to \cite[Ch.\ 1]{Lawson-Michelsohn89} or \cite[Ch.\ 1]{Baum81}. 

We view $\R^n$ as a subset of $\R^{n+1}$. We consider a basis $\{e_0,\ldots,e_n\}$ of $\R^{n+1}$, where $\{e_1,\ldots,e_n\}$ is a basis of the embedded $\R^n$. 
We distinguish two cases: either $\R^{n+1}$ is Euclidean space, or $\R^{n+1}=\R^{1,n}$ is Minkowski space. In both cases we write $\la\cdot|\cdot\ra$ for the corresponding scalar product on $\R^{n+1}$. 
Let $\sigma$ denote the signature of $\R^{n+1}$, i.e.\ $\sigma = (0,n+1)$ in the Euclidean case, and $\sigma = (1,n)$ in the Minkowski case. 
In the latter case, we have $\la e_0|e_0\ra=-1$ and $\la e_j|e_j\ra=1$ for $j=1,\ldots,n$. In either case, we write $\epsilon_0 := \la e_0|e_0\ra = \pm 1$. 
Furthermore, we write $\tau_0:=1$ if $\epsilon_0=1$, and $\tau_0:=i$ if $\epsilon_0=-1$ (in both cases we have $\tau_0^2=\epsilon_0$). 

The real Clifford algebra $\Cliff_\sigma$ is the real algebra generated by $v,w\in\R^{n+1}$ modulo the relation $vw+wv=-2\la v|w\ra$. The complex Clifford algebra is the complexification $\CCliff_{n+1} := \Cliff_\sigma\otimes\C$ (which is independent of the signature). 
We write $\CCliff_{n+1}^0$ for the even subalgebra generated by products $vw$, for $v,w\in\R^n$. 

If $n+1=2m$ is even, there is a unique irreducible representation (up to equivalence) of the complex Clifford algebra $\CCliff_{n+1}$ on the representation space $\Delta_{2m} := \C^{2^m}$, which we write as $\Phi_{n+1}\colon\CCliff_{n+1}\xrightarrow{\simeq}M_{2^m}(\C)$. 
If $n+1=2m+1$ is odd, there are two inequivalent irreducible representations of $\CCliff_{n+1}$ on the representation space $\Delta_{2m+1} := \C^{2^m}$, which we write as $\Phi_{n+1}^\pm\colon\CCliff_{n+1}\to M_{2^m}(\C)$. 
In this case, we have the isomorphism $\Phi_{n+1}^+\oplus\Phi_{n+1}^-\colon \CCliff_{n+1}\xrightarrow{\simeq}M_{2^m}(\C)\oplus M_{2^m}(\C)$. 
We can choose these representations such that $\Phi_{n+1}^-(w) = -\Phi_{n+1}^+(w)$, for any $w\in\R^{n+1}\subset\CCliff_{n+1}$. 
We write $\hat\Phi_{n+1} := \Phi_{n+1}^+$. 

The spin group is defined as the group whose elements are products of an even number of unit vectors:
\[
\Spin_\sigma := \left\{ v_1\cdots v_{2k} \in\Cliff_\sigma^0 \mathop{\big|} q(v_j,v_j)=\pm1,\, 1\leq 2k\leq n+1,\, 1\leq j\leq 2k \right\} .
\]
The spin group $\Spin_\sigma$ is a double cover of the special pseudo-orthogonal group $\SO_\sigma$ via a homomorphism $\lambda\colon\Spin_\sigma\to\SO_\sigma$. The subgroup $\SO^+_\sigma\subset\SO_\sigma$, given by the connected component of the identity, corresponds to all orthogonal transformations which preserve both space- and time-orientation. We define $\Spin^+_\sigma$ as its pre-image under $\lambda$:
\[
\Spin^+_\sigma := \lambda^{-1}\big(\SO^+_\sigma\big) = \Big\{ v_1\cdots v_{2k} \in\Spin_\sigma : \prod_{j=1}^{2k} q(v_j,v_j) = 1 \Big\} .
\]
In the Euclidean case, $\SO_{n+1}$ is connected, and we simply have $\Spin^+_{n+1} = \Spin_{n+1}$. 

From here on we consider the case of odd $n=2m+1$. We will make a specific choice for our standard representation of $\CCliff_{n+1}$ on $\Delta_{n+1} = \Delta_n \oplus \Delta_n$. 
Given the representation $\hat\Phi_n\colon\CCliff_n\to M_{2^m}(\C)$, we define the representation $\hat\Phi_{n+1}\colon\CCliff_{n+1}\to M_{2^{m+1}}(\C)$ by 
\begin{align}
\label{eq:Cliff_rep}
\hat\Phi_{n+1}(w) &:= \mattwo{0}{i\hat\Phi_n(w)}{-i\hat\Phi_n(w)}{0} , & 
\hat\Phi_{n+1}(e_0) &:= \mattwo{0}{i\tau_0\epsilon_0}{i\tau_0\epsilon_0}{0} ,
\end{align}
where $w\in\R^n\subset\R^{n+1}$. 
The representation $\hat\Phi_{n+1}$ is (up to equivalence) the unique irreducible representation of $\CCliff_{n+1}$. 
Its restriction $\hat\Phi_{n+1}\colon\Spin^+_\sigma\to M_{2^{m+1}}(\C)$ decomposes as the direct sum of two equivalent irreducible representations of the spin group $\Spin^+_\sigma$. 

We have the isomorphism $\varphi\colon\CCliff_n \to \CCliff^0_{n+1}$ given by $e_j \mapsto \tau_0e_0e_j$ (see \cite[Theorem 3.7]{Lawson-Michelsohn89}).
Then we also obtain a representation of $\CCliff_n$ on $\Delta_{n+1}$ by setting $\til\Phi_n := \hat\Phi_{n+1} \circ \varphi$. For $w\in\R^n$ we then have 
\begin{align}
\label{eq:Cliff_rep_hs}
\til\Phi_n(w) := \tau_0 \hat\Phi_{n+1}(e_0) \hat\Phi_{n+1}(w) = \mattwo{\hat\Phi_n(w)}{0}{0}{-\hat\Phi_n(w)} .
\end{align}

Using the standard basis of $\Delta_{n+1} = \C^{2^{m+1}}$, we define a positive-definite inner product
$$
\la v,w\ra^\pos := \sum_{j=1}^{2^m} \bar{v_j} w_j ,
$$
for $v,w\in\Delta_{n+1}$. In the Euclidean case, this inner product is invariant under the action of the spin group $\Spin_{n+1}$. 
In the Lorentzian case however, this inner product is only invariant under the action of the maximal compact subgroup $\Spin_n \subset \Spin^+_{n+1}$. 
In this case, there exists instead a canonical \emph{indefinite} inner product $\la\cdot,\cdot\ra$ on $\Delta_{n+1}$ which is invariant under $\Spin^+_{n+1}$ \cite[Satz 1.12]{Baum81}. Given the basis $\{e_0,\ldots,e_n\}$ of $\R^{n+1}$, the canonical inner product is related to the positive-definite inner product via
\[
\la v,w\ra = \la v,\hat\Phi_{n+1}(e_0)w\ra^\pos . 
\]

\subsection{Spin geometry}
\label{sec:spin_structures}

Usually, a spin structure is defined as a double cover of the bundle of oriented \emph{(pseudo-)orthonormal} frames, and therefore it relies on the (pseudo-)Riemannian metric. Since we will be interested in studying a family of metrics, it will be more convenient to use a topological definition of spin structures which is independent of the metric (this idea was first suggested by Milnor in \cite{Mil65}, and is described in detail in \cite{DP86}). In this section we will introduce both metric and topological spin structures. 

Let $X$ be an oriented smooth manifold of dimension $n+1\geq3$. We denote by $\GL^+(X)\to X$ the principal $\GL^+_{n+1}$-bundle of oriented frames on $X$. Consider the universal double cover $\tau\colon\til\GL^+_{n+1}\to\GL^+_{n+1}$. 

\begin{defn}[{\cite[\S2]{DP86}}]
A \emph{topological spin structure} on $X$ is given by a principal $\til\GL^+_{n+1}$-bundle $\til\GL^+(X)\to X$ and a principal bundle morphism $\eta\colon\til\GL^+(X)\to\GL^+(X)$, such that the following diagram commutes.
\begin{align*}
\xymatrix{
\til\GL^+(X) \times \til\GL^+_{n+1} \ar[r] \ar[dd]^{\eta\times\tau} & \til\GL^+(X) \ar@{->>}[dr] \ar[dd]^{\eta} & \\
 & & X \\
\GL^+(X) \times \GL^+_{n+1} \ar[r] & \GL^+(X) \ar@{->>}[ur] & \\
}
\end{align*}
\end{defn}

Now let $(X,g)$ be a (space- and time-oriented) pseudo-Riemannian manifold, either of Riemannian signature $\sigma=(0,n+1)$ or of Lorentzian signature $\sigma=(1,n)$. We denote by $\SO^+_g(X)\to X$ the principal $\SO^+_\sigma$-bundle of space- and time-oriented pseudo-orthonormal frames on $X$. Denote by $\lambda\colon\Spin^+_\sigma\to \SO^+_\sigma$ the double cover. 

\begin{defn}
A \emph{metric spin structure} on $(X,g)$ is given by a principal $\Spin^+_\sigma$-bundle $\Spin^+_g(X) \to X$ and a principal bundle morphism $\eta_g\colon\Spin^+_g(X)\to\SO^+_g(X)$, such that the following diagram commutes. 
\begin{align*}
\xymatrix{
\Spin^+_g(X) \times \Spin^+_\sigma \ar[r] \ar[dd]^{\eta_g\times\lambda} & \Spin^+_g(X) \ar@{->>}[dr] \ar[dd]^{\eta_g} & \\
 & & X \\
\SO^+_g(X) \times \SO^+_\sigma \ar[r] & \SO^+_g(X) \ar@{->>}[ur] & \\
}
\end{align*}
\end{defn}

Given a pseudo-Riemannian metric $g$ on $X$, a topological spin structure $(\til\GL^+(X),\eta)$ can be restricted to a metric spin structure $(\Spin^+_g(X),\eta_g)$ by setting $\Spin^+_g(X) := \eta^{-1}(\SO^+_g(X))$ and $\eta_g := \eta|_{\eta^{-1}(\SO^+_g(X))}$. Conversely, any metric spin structure can be extended to a topological spin structure by setting $\til\GL^+(X) := \Spin^+(X) \times_{\Spin^+_\sigma} \til\GL^+_{n+1}$ and $\eta := \eta_g \times \tau$. Both restriction and extension preserve the notions of equivalence of spin structures \cite[\S2]{DP86}. 

The spinor bundle on $(X,g)$ is given by the associated vector bundle
\begin{align*}
\bS_X &:= \Spin^+_g(X) \times_{\hat\Phi_{n+1}} \Delta_{n+1} , 
\end{align*}
where $\Delta_{n+1}$ is the standard representation space of the spin group $\Spin^+_\sigma$. 
We point out that here we cannot use the topological spin structure, because the fundamental spin  representation does not lift to a representation of $\til\GL^+_{n+1}$.

The tangent bundle can be viewed as the associated vector bundle $TX = \SO_g^+(X) \times_{\SO^+_\sigma} \R^{n+1}$. 
We define the \emph{Clifford bundle} as the associated bundle
$$
\Cliff(TX,g) := \SO_g^+(X) \times_{\SO^+_\sigma} \Cliff_\sigma .
$$
Given the spin structure $\Spin^+(X)$, we can write $\SO_g^+(X) = \Spin^+_g(X) \times_\lambda \SO^+_\sigma$, where $\lambda$ is the double cover $\Spin^+_\sigma\to \SO^+_\sigma$. We can then also view the Clifford bundle as an associated bundle of $\Spin^+(X)$ via
$$
\Cliff(TX,g) = \Spin^+(X) \times_\lambda \SO^+_\sigma \times_{\SO^+_\sigma} \Cliff_\sigma = \Spin^+_g(X) \times_{\Ad} \Cliff_\sigma ,
$$
where $\Ad$ is given by $\Ad_u(a) = u\cdot a\cdot u^{-1}$ for all $u\in\Spin^+_\sigma$ and $a\in\Cliff_\sigma$. 
The complexified Clifford bundle is independent of the signature of $g$ and is denoted $\CCliff(TX) = \Cliff(TX,g) \otimes_\R \C = \SO_g^+(X) \times_{\SO^+_\sigma} \CCliff_{n+1}$. 

Using the natural inclusion $\iota\colon\R^{n+1}\into\Cliff_\sigma$, we can define the Clifford representation $\gamma_X\colon TX \into \Cliff(TX,g)$ by 
$$
\gamma_X([f,x]) := [f,\iota(x)] ,
$$
where $f\in\SO_g^+(X)$ and $x\in\R^{n+1}$ determine $[f,x]\in TX = \SO_g^+(X) \times_{\SO^+_\sigma} \R^{n+1}$. 
This Clifford representation inherits the Clifford relation of $\Cliff_\sigma$, so we have $\gamma_X(v)\gamma_X(w) + \gamma_X(w)\gamma_X(v) = -2g(v,w)$ for all $v,w\in TX$. 

The canonical inner product on $\Delta_{n+1}$ yields a canonical hermitain structure 
\[
(\cdot|\cdot)\colon \Gamma_c^\infty(\bS_X) \times \Gamma_c^\infty(\bS_X) \to C_c^\infty(X) ,
\]
which gives rise to the inner product
$
\la\psi_1|\psi_2\ra := \int_X (\psi_1|\psi_2) \dvol_g ,
$
for all $\psi_1,\psi_2\in\Gamma_c^\infty(X,\bS_X)$, where $\dvol_g$ 
denotes the canonical volume form of $(X,g)$. 
In the Riemannian case, this inner product is positive-definite, and we write $\la\cdot|\cdot\ra^\pos = \la\cdot|\cdot\ra$. 
In the Lorentzian case, this inner product is indefinite (but non-degenerate). However, given a global unit timelike vector field $\nu$ (which exists because $X$ is space- and time-oriented), we obtain a positive-definite Hermitian structure
\[
(\cdot|\cdot)^\pos := (\cdot|\gamma_X(e_0)\cdot) ,
\]
yielding a positive-definite inner product
$
\la\psi_1|\psi_2\ra^\pos := \int_X (\psi_1|\psi_2)^\pos \dvol_g .
$
The completion of $\Gamma_c^\infty(X,\bS_X)$ with respect to $\la\cdot|\cdot\ra^\pos$ is denoted $L^2(X,\bS_X)$. 
In the Lorentzian case, $L^2(X,\bS_X)$ is a Krein space with fundamental symmetry $\mJ_X = \gamma_X(e_0)$ (for more information on Krein spaces, see our summary in \cref{sec:Krein}, or refer to \cite{Bognar74} for a detailed introduction).

Locally, we can write a spinor $\psi\in\Gamma_c^\infty(\bS_X)$ as the equivalence class $[s,v]$, where $s$ is a local section of $\Spin^+_g(X)$ and $v$ is a local function with values in $\Delta_{n+1}$. The double cover $\eta_g\colon \Spin^+_g(X) \to \SO^+_g(X)$ then yields a local \mbox{(pseudo-)orthonormal} frame $\eta_g(s) = \{e_j\}$, such that 
$g(e_i,e_j) = \delta_{ij}\epsilon_j$ (where $\epsilon_j=1$ for $j\neq0$). 
The Levi-Civita connection on the tangent bundle lifts to a connection on the spinor bundle. Locally, this spin connection takes the form (see \cite[Satz 3.2]{Baum81} and \cite[Eq.(2.5)]{BGM05})
\begin{align}
\label{eq:spin_connection}
\nabla^{\bS_X}_Y\psi = \bigg[ s , Y(v) + \frac12 \sum_{j<k} \epsilon_j \epsilon_k g(\nabla_Ye_j,e_k) \gamma_X(e_j) \gamma_X(e_k) v \bigg] ,
\end{align}
for a local vector field $Y = \sum_j Y^j e_j \in\Gamma_c^\infty(TX)$. 

The Dirac operator $\sD_X$, canonically associated to the metric $g$, is defined as
\begin{align}
\label{eq:Dirac}
\sD_X &:= \sum_{j=0}^n \epsilon_j \gamma_X(e_j) \nabla^{\bS_X}_{e_j} .
\end{align}
In the Riemannian case, the Dirac operator $\sD_X$ is symmetric \cite[Proposition II.5.3]{Lawson-Michelsohn89}. 
In the Lorentzian case, the operator $i\sD_X$ is Krein-symmetric \cite[Satz 3.18]{Baum81} (i.e., it is symmetric with respect to the canonical indefinite inner product).

\subsection{A hypersurface}

In this section we will describe the spin geometry of an embedded hypersurface. 
The Dirac operator on a hypersurface in flat Euclidean space was already studied in \cite{Tra92}. 
Here we largely follow the general exposition for hypersurfaces in pseudo-Riemannian manifolds given in \cite[\S3]{BGM05}. 

Let $(Z,g)$ be an oriented spin manifold of dimension $n+1$, where the metric $g$ is either Riemannian or Lorentzian. 
We will denote the signature of $Z$ by $\sigma$, i.e.\ $\sigma=(0,n+1)$ if $g$ is Riemannian, or $\sigma=(1,n)$ if $g$ is Lorentzian. 
In the Lorentzian case, we assume furthermore that $Z$ is also time-oriented. 
We assume that $Z$ comes equipped with a given topological spin structure $\eta\colon\til\GL^+(Z)\to\GL^+(Z)$, with corresponding metric spin structure $\eta_g\colon\Spin^+_g(Z)\to\SO^+_g(Z)$. 

We will consider a codimension $1$ hypersurface $M\subset Z$ with \emph{trivial} normal bundle. This means there is a vector field $\nu=e_0$ on $Z$ along $M$ satisfying $\epsilon_0 := g(\nu,\nu) = \pm1$ and $g(\nu,TM)=0$. If $Z$ is Lorentzian, we assume that the vector field $\nu$ is timelike. Thus the induced metric $g_M$ on $M$ is positive-definite. 

The hypersurface $M$ inherits a spin structure from $Z$ via the decomposition $TZ|_M = \R\oplus TM$ given by $\nu$, as follows. 
The bundle of oriented frames $\GL^+(M)$ on $M$ can be embedded into the bundle of (space- and \mbox{time-)oriented} frames $\GL^+(Z)|_M$ of $Z$ restricted to $M$ by the map $\iota\colon (e_1,\ldots,e_n) \mapsto (\nu,e_1,\ldots,e_n)$. Similarly, the bundle of oriented \mbox{(pseudo-)orthonormal} frames $\SO^+_{g_M}(M)$ on $M$ can be embedded into the bundle of (space- and time-)oriented \mbox{(pseudo-)orthonormal} frames $\SO^+_g(Z)|_M$ of $Z$ restricted to $M$ by the map $\iota_g := \iota|_{\SO^+_{g_M}(M)}$. Then 
\begin{align}
\label{eq:spin_restriction}
\til\GL^+(M) &:= \eta^{-1}(\iota(\GL^+(M))) , & \Spin^+_{g_M}(M) &:= \eta_g^{-1}(\iota_g(\SO^+_{g_M}(M)))
\end{align}
define the topological and metric spin structures on $M$. It is clear that $\Spin^+_{g_M}(M)$ is identical to the metric spin structure obtained by restricting $\til\GL^+(M)$ using the metric $g_M$ (as described in \cref{sec:spin_structures}). 

\begin{assumption}
We will assume throughout this article that $Z$ is \emph{even-dimensional}, so that $M$ is odd-dimensional. 
\end{assumption}

Recall from \cref{sec:Clifford_alg} that there is then a unique irreducible representation $\hat\Phi_{n+1}$ of $\CCliff_{n+1}$ on $\Delta_{n+1}$, while there are two inequivalent irreducible representations $\Phi_n^\pm=\pm\hat\Phi_n$ of $\CCliff_{n}$ on $\Delta_{n}$. 
The spinor bundles on $Z$ and $M$ are given by the associated vector bundles 
\begin{align*}
\bS_Z &:= \Spin^+_g(Z) \times_{\hat\Phi_{n+1}} \Delta_{n+1} , & \bS_M^\pm &:= \Spin_g(M) \times_{\pm\hat\Phi_n} \Delta_n .
\end{align*}
We recall that, though the representations $\Phi_n^\pm=\pm\hat\Phi_n$ are inequivalent as representations of $\CCliff_n$, they are \emph{equivalent} as representations of $\Spin_n$. Hence the spinor bundles $\bS_M^+$ and $\bS_M^-$ are isomorphic. Furthermore, by our definition of $\hat\Phi_{n+1}$ in \cref{eq:Cliff_rep}, 
we have that $\bS_Z|_M = \bS_M^+\oplus\bS_M^-$, and this direct sum decomposition is precisely the decomposition corresponding to the $\Z_2$-grading on $\bS_Z$. Using \cref{eq:Cliff_rep}, we see that the Clifford representation is given by
\begin{align}
\label{eq:Clifford_rep}
\gamma_Z(X) &= \mattwo{0}{i\gamma_M(X)}{-i\gamma_M(X)}{0} , & 
\gamma_Z(\nu) &= \mattwo{0}{i\tau_0\epsilon_0}{i\tau_0\epsilon_0}{0} ,
\end{align}
where $X$ is a vector field on $M$. 
Furthermore, using the representation $\til\Phi$ from \cref{eq:Cliff_rep_hs}, we also have
\[
\til\gamma_M(X) := \tau_0 \gamma_Z(\nu) \gamma_Z(X) = \mattwo{\gamma_M(X)}{0}{0}{-\gamma_M(X)} ,
\]
which provides a representation of the Clifford algebra of $M$ on $\bS_Z|_M=\bS_M^+\oplus\bS_M^-$. 

\begin{remark}
The case in which $Z$ is odd-dimensional requires a separate treatment, because in this case the spinor bundles on $Z$ and $M$ have the same rank. Hence we do not have the decomposition $\bS_Z|_M = \bS_M^+\oplus\bS_M^-$, and the analogue of \cref{eq:Clifford_rep} would be rather different. 
In this article we focus only on even dimensions, because it is our aim to describe $(3+1)$-dimensional spacetime. 
\end{remark}

The Levi-Civita connections on the pseudo-Riemannian manifold $(Z,g)$ and the hypersurface $(M,g_M)$ are denoted by $\nabla^Z$ and $\nabla^M$, respectively. When restricted to the hypersurface $M$, they are related via $\nabla^M_X = P_M \circ \nabla^Z_X$, where $P_M \colon TZ\to TM$ denotes the orthogonal projection (see \cite[\S3.5]{BEE96}), and their difference determines the \emph{Weingarten map} $W\colon TM\to TM$ with respect to $\nu$, or the \emph{second fundamental form} $K\colon TM\times TM\to\R$, given by
$$
g(W(X),Y) \nu := K(X,Y) \nu := \nabla^Z_X Y - \nabla^M_X Y 
$$
for all vector fields $X$ and $Y$ on $M$. 
Since the connections are torsion-free, it follows that the second fundamental form is symmetric, i.e.\ $K(X,Y) = K(Y,X)$. Using also the metric compatibility of the connection, the Weingarten map can explicitly be obtained as
\begin{align}
\label{eq:Weingarten}
W(X) = - \epsilon_0 \nabla^Z_X \nu .
\end{align}
Thus, the Weingarten map describes how the normal field $\nu$ changes along the surface $M$, and as such it describes the extrinsic curvature of $M$. For Clifford multiplication with $W(X)$ we can write
\begin{align}
\label{eq:tr_Weingarten}
\gamma_Z(W(X)) &= \sum_{j=1}^n g(W(X),e_j) \gamma_Z(e_j) , & \sum_{j=1}^n \gamma_Z(e_j) \gamma_Z(W(e_j)) &= - \tr^M(W) ,
\end{align}
where $\tr^M(W) := \tr^M(K) = \sum_{j=1}^n K(e_j,e_j)$. 

The Levi-Civita connections $\nabla^Z$ and $\nabla^M$ can be lifted to connections $\nabla^{\bS_Z}$ and $\nabla^{\bS_M}$ on the spinor bundles $\bS_Z$ and $\bS_M$, respectively, and are given explicitly in \cref{eq:spin_connection}. 
Given a section $\psi = [s,w]$ of the spinor bundle $\bS_Z|_M$, we can use \cref{eq:Weingarten,eq:tr_Weingarten} to rewrite the spin connection on $\bS_Z$ as
\begin{align*}
\nabla^{\bS_Z}_{e_l} \psi &= \Big[ s , e_l(w) + \frac12 \sum_{1\leq j<k\leq n} g(\nabla^Z_{e_l}e_j,e_k) \gamma_Z(e_j) \gamma_Z(e_k) w + \frac12 \sum_{1\leq k\leq n} \epsilon_0 g(\nabla^Z_{e_l}\nu,e_k) \gamma_Z(\nu) \gamma_Z(e_k) w \Big] \\
&= \Big[ s , e_l(w) + \frac12 \sum_{1\leq j<k\leq n} g(\nabla^M_{e_l}e_j,e_k) \gamma_Z(e_j) \gamma_Z(e_k) w - \frac12 \sum_{1\leq k\leq n} g(W(e_l),e_k) \gamma_Z(\nu) \gamma_Z(e_k) w \Big] \\
&= \nabla^{\bS_M}_{e_l} \psi - \frac12 \gamma_Z(\nu) \gamma_Z(W(e_l)) \psi .
\end{align*}
Hence for any vector field $X$ on $M$ we have (cf.\ \cite[Eq.\ (3.5)]{BGM05})
\begin{align}
\label{eq:spin_conn_Weingarten}
\nabla^{\bS_Z}_X = \nabla^{\bS_M}_X - \frac12 \gamma_Z(\nu) \gamma_Z(W(X)) .
\end{align}

The Dirac operators $\sD_Z$ and $\sD_M$, canonically associated to the metrics $g$ and $g|_M$ (respectively), are defined (see \cref{eq:Dirac}) as
\begin{align*}
\sD_Z &:= \epsilon_0 \gamma_Z(\nu) \nabla^{\bS_Z}_{\nu} + \sum_{j=1}^n \gamma_Z(e_j) \nabla^{\bS_Z}_{e_j} , & 
\sD_M &:= \sum_{j=1}^n \gamma_M(e_j) \nabla^{\bS_M}_{e_j} .
\end{align*}
Using \cref{eq:spin_conn_Weingarten} we have 
\begin{align*}
\sD_Z|_M 
&= \epsilon_0 \gamma_Z(\nu) \nabla^{\bS_Z}_{\nu} + \sum_{j=1}^n \gamma_Z(e_j) \nabla^{\bS_M}_{e_j} - \frac12 \sum_{j=1}^n \gamma_Z(e_j) \gamma_Z(\nu) \gamma_Z(W(e_j)) \\
&= \epsilon_0 \gamma_Z(\nu) \nabla^{\bS_Z}_{\nu} - \tau_0 \gamma_Z(\nu) \sum_{j=1}^n \til\gamma_M(e_j) \nabla^{\bS_M}_{e_j} + \frac12 \gamma_Z(\nu) \sum_{j=1}^n \gamma_Z(e_j) \gamma_Z(W(e_j)) ,
\end{align*}
where we have used that $\til\gamma_M(e_j) = \tau_0 \gamma_Z(\nu) \gamma_Z(e_j)$. 
We will write $\til\sD_M = \mattwo{\sD_M}{0}{0}{-\sD_M}$. 
From \cref{eq:tr_Weingarten} we see that the mean curvature of $M$ is given by $H := \frac1n \tr^M(W) = -\frac1n \sum_j \gamma_Z(e_j) \gamma_Z(W(e_j))$. 
The Dirac operators on $Z$ and $M$ are then related via (cf.\ \cite[Eq.\ (3.6)]{BGM05})
\begin{align}
\label{eq:Dirac_fol}
\sD_Z|_M &= \epsilon_0 \gamma_Z(\nu) \nabla^{\bS_Z}_{\nu} - \tau_0 \gamma_Z(\nu) \til\sD_M - \frac n2 H \gamma_Z(\nu) .
\end{align}

\subsubsection{The spectral triple}

Using the positive-definite Hermitian structure $(\cdot|\cdot)^\pos$ on $\bS_Z$, we define the Hilbert space $L^2(M,\bS_Z|_M)$ of square-integrable spinors on $M$ as the completion of $\Gamma_c^\infty(M,\bS_Z|_M)$ with respect to the positive-definite inner product 
\begin{align}
\label{eq:pos_inner_prod}
\la\phi|\psi\ra_M^\pos := \int_M (\phi|\psi)^\pos \dvol_M , 
\end{align}
where $\dvol_M = \nu \contract \dvol_Z$ is the volume form on $M$ induced by the volume form on $Z$. 
Since the Hermitian structure $(\cdot|\cdot)^\pos$ on $\bS_Z = \Spin^+_g(Z) \times_{\Spin^+_\sigma} \Delta_{n+1}$ is obtained from the standard positive-definite inner product on $\Delta_{n+1}$, we note that the decomposition $\bS_Z|_M = \bS_M^+\oplus\bS_M^-$ is an orthogonal direct sum, and the Hermitian structure on $\bS_Z|_M$ agrees with the intrinsic Hermitian structures on $\bS_M^\pm$. Hence we have the isomorphism $L^2(M,\bS_Z|_M) \simeq L^2(M,\bS_M^+) \oplus L^2(M,\bS_M^-)$. 

We consider the Dirac operators $\pm\sD_M$ which are canonically associated to the spinor bundles $\bS_M^\pm$ and the Riemannian metric $g_M$. 
The following statement is well-known, and we refer to e.g.\ \cite[Ch.\ 10]{Higson-Roe00} for more details. 
\begin{prop}
\label{prop:ST_t}
If the metric $g_M$ is \emph{complete}, we obtain spectral triples $\big( C_c^\infty(M) , L^2(M,\bS_M^\pm) , \pm\sD_M \big)$. 
\end{prop}

\section{Product space(time)s}
\label{sec:prod_spacetime}

As before, we consider both the Riemannian and the Lorentzian case. If $(Z,g)$ is an oriented Riemannian manifold, we will refer to $(Z,g)$ as a \emph{space}. 
If $(Z,g)$ is a time-oriented Lorentzian manifold, we call $(Z,g)$ a \emph{spacetime}. A \emph{temporal function} on $(Z,g)$ is a smooth function $T\colon Z\to\R$ such that the gradient $\nabla T$ is timelike and past-directed everywhere. A spacetime $(Z,g)$ admits a temporal function if and only if it is stably causal \cite{BS05}. 
We will restrict our attention to space(time)s which admit a smooth orthogonal splitting, as follows. 

\begin{defn}
\label{defn:product}
A space(time) is called a \emph{product space(time)} if it is isometric to a space(time) of the form $(M\times\R,g_\bullet + \epsilon_0 N^2dT^2)$, where 
$\epsilon_0=1$ for a space and $\epsilon_0=-1$ for a spacetime,  
$M$ is a smooth (spacelike) hypersurface, 
$N\colon M\times\R\to(0,\infty)$ is a smooth positive function, 
$T\colon M\times\R\to\R$ is the canonical projection, 
and $g_\bullet = \{g_t\}_{t\in\R}$ is a smooth family of Riemannian metrics on $M$. 
\end{defn}

The term \emph{product space(time)} obviously refers to the fact that $Z$ is the topological product of the hypersurface $M$ and the real line $\R$. However, we emphasise that $Z$ is \emph{not} a geometric product of $M$ and $\R$; indeed, the metric $g$ is allowed to be much more general than a product metric of the form $g_0+\epsilon_0dT^2$. 

We will refer to the smooth function $N$ as the \emph{lapse function}, which is standard terminology in the Lorentzian case. 
We will often think of $N$ as a smooth family of strictly positive smooth functions $N_\bullet = \{N_t\}_{t\in\R}\subset C^\infty(M)$, by setting $N_t(x) := N(x,t)$ for $x\in M$. 

The function $T$ defines a global coordinate on $M$, which we refer to as the time coordinate. 
Consider local coordinates on $Z=M\times\R$ given by local coordinates on $M$ along with the global time coordinate $T$ on $\R$. 
We consider the unit normal vector field $\nu = e_0$. Then $\epsilon_0 = g(\nu,\nu) = \epsilon_0 N^2 dT(\nu)^2$, so $dT(\nu) = N^{-1}$. Since $dT(\partial_T) = \partial_T(T) = 1$, we see that $\nu = N^{-1} \partial_T$. Noting that $g(\partial_T,\partial_T) = \epsilon_0 N^2 dT(\partial_T)^2 = \epsilon_0 N^2$ and $g(\partial_T,\nabla T) = dT(\partial_T) = 1$, we also see that $\nabla T = \epsilon_0 N^{-2} \partial_T = \epsilon_0 N^{-1} \nu$. The lapse function $N$ can then be written as 
\[
N = \big( \epsilon_0 g(\partial_T,\partial_T) \big)^{\frac12} = \big( \epsilon_0 g(\nabla T,\nabla T) \big)^{-\frac12} . 
\]
In the Lorentzian case, $\nabla T$ is timelike, so the smooth function $T$ is a temporal function. Hence every product spacetime is stably causal. 
The converse need not be true. However, let $(Z,g)$ be stably causal with a temporal function $T$ on $Z$, and consider the corresponding rescaled (conformally equivalent) metric $g_c := -g(\nabla T,\nabla T) g$. If $(Z,g_c)$ is timelike geodesically complete, then by \cite[Theorem 7.3.4]{GK99} it follows that $(Z,g)$ is a product spacetime. 

A spacetime is called globally hyperbolic if there exists a Cauchy hypersurface, i.e.\ a hypersurface $M\subset Z$ which is intersected exactly once by any inextendible timelike curve. 
A temporal function is called \emph{Cauchy} if every level set $M_t := \{x\in Z : T(x)=t\}$ is a Cauchy hypersurface of $Z$. Hence the existence of a Cauchy temporal function on $Z$ implies that $Z$ is globally hyperbolic. Conversely, it was shown in \cite{BS05} that every globally hyperbolic spacetime admits a Cauchy temporal function. Furthermore, since the level sets of a Cauchy temporal function are Cauchy hypersurfaces, and since all Cauchy hypersurfaces must be diffeomorphic, this Cauchy temporal function determines a splitting as in \cref{defn:product} (see \cite[Theorem 1.1]{BS05} for details). Thus, every globally hyperbolic spacetime is a product spacetime, with the additional property that $M$ and $T$ can be taken to be Cauchy. Finally, we also have the following sufficient conditions for when a product spacetime is in fact globally hyperbolic. 

\begin{thm}[{\cite[Theorem 2.1]{CC02}}]
\label{thm:product_glob_hyp}
Consider a product spacetime $(Z,g) = (M\times\R,g_\bullet + \epsilon_0 N^2dT^2)$ satisfying the following assumptions: 
\begin{enumerate}
\item there exist positive numbers $N_1,N_2>0$ such that $N_1 < N(x,t) < N_2$ for all $(x,t)\in M\times\R$;
\item the Riemannian metrics $g_t$ on $M\times\{t\}$ are complete, and uniformly bounded below for all $t\in\R$ by some complete metric $h$ on $M$. 
\end{enumerate}
Then $(Z,g)$ is globally hyperbolic. 
\end{thm}
\begin{remark}
The assumptions in the above theorem are not needed for this section. However, most of these assumptions will be relevant for the abstract description in terms of families of spectral triples in \cref{sec:prod_spec_trip}. 
\end{remark}

\subsection{Spin structures}

Suppose that a product space(time) $(Z,g) = (M\times\R,g_\bullet + \epsilon_0 N^2dT^2)$ is equipped with a given 
topological spin structure $\til\GL^+(Z)$, and let $\Spin^+_g(Z)$ be the corresponding metric spin structure. 
We consider $M_t := (M\times\{t\},g_t)$ as a Riemannian submanifold of $Z$. 
As in \cref{eq:spin_restriction}, we obtain a topological spin structure $\til\GL^+(M_t)$ and a metric spin structure $\Spin_{g_t}(M_t)$ for each $t$, and we will always consider $M_t$ to be equipped with these spin structures. As mentioned after \cref{eq:spin_restriction}, the topological spin structure obtained by extending the structure group of $\Spin_{g_t}(M_t)$ to $\til\GL^+_n$ is identical to $\til\GL^+(M_t)$. 
Although the topological spin structure of $M_t$ is independent of the metric $g_t$, we emphasise that (in principle) it still depends on $t$ through the inclusion $M_t = M\times\{t\} \subset Z$. 
However, the topological spin structures $\til\GL^+(M_t)$ on $M$ are all equivalent. 
Indeed, for $t_0<t_1\in\R$, we can identify the fibres of $\til\GL^+(M_{t_0})$ and $\til\GL^+(M_{t_1})$ by parallel transport along the path $t\to(x,t)$ for $t\in[t_0,t_1]$. Since parallel transport on a principal bundle is compatible with the right action of the structure group, this yields a principal bundle isomorphism $\til\GL^+(M_{t_0}) \to \til\GL^+(M_{t_1})$. 

Conversely, suppose we have a smooth manifold $M$ with a topological spin structure $\til\GL^+(M)$, a smooth family $N_\bullet = \{N_t\}_{t\in\R}$ of strictly positive smooth functions on $M$, and a smooth family $g_\bullet = \{g_t\}_{t\in\R}$ of Riemannian metrics on $M$. 
We consider the product space(time) $(Z,g) := (M\times\R,g_\bullet + \epsilon_0 N^2dT^2)$. 
The spin structure on $Z$ can be reconstructed as follows.
Denote by $\pi\colon M\times\R\to M$ the canonical projection. Then the pullback bundle $\pi^*(\til\GL^+(M))$ is a principal $\til\GL^+_n$-bundle over $Z = M\times\R$. By enlarging the structure group, we can extend this pullback bundle to a principal $\til\GL^+_{n+1}$-bundle.

Now consider a product space(time) $(Z,g) := (M\times\R,g_\bullet + \epsilon_0 N^2dT^2)$ with given topological spin structure $\til\GL^+(Z)$. Let $\til\GL^+(M_0)$ be the topological spin structure on $M_0=M\times\{0\}$ obtained from \cref{eq:spin_restriction}, and consider the pullback bundle $\pi^*(\til\GL^+(M_0))$ over $Z = M\times\R$. 
We can reduce the principal $\til\GL^+_{n+1}$-bundle $\til\GL^+(Z)$ to a principal $\til\GL^+_n$-bundle $\bP\to Z$ by defining the fibres as
$$
\bP_{(x,t)} := \eta^{-1}(\iota(\GL^+(M_t)_x)) .
$$
Then $\til\GL^+(Z) = \bP \times_{\til\GL^+_n} \til\GL^+_{n+1}$. Since $\bP|_{M_t} = \eta^{-1}(\iota(\GL^+(M_t))) = \til\GL^+(M_t) \simeq \til\GL^+(M_0)$, we see that $\bP \simeq \pi^*(\til\GL^+(M_0))$. 
Therefore we have the equivalence
$$
\til\GL^+(Z) \simeq \pi^*(\til\GL^+(M_0)) \times_{\til\GL^+_n} \til\GL^+_{n+1} ,
$$
showing that the spin structure on $Z$ can be reconstructed from the spin structure on $M_0$ (up to equivalence). 
Thus, up to equivalence of the spin structures, we have explicitly constructed a bijection between product spin space(time)s $(Z,g,\til\GL^+(Z))$ and the `foliation data' $(M,g_\bullet,N_\bullet,\til\GL^+(M))$.

\subsection{Parallel transport}

Consider a product space(time) $(Z,g) = (M\times\R,g_\bullet + \epsilon_0 N^2dT^2)$ equipped with a given topological spin structure $\til\GL^+(Z)$.
On each hypersurface $M_t = M\times\{t\}$ we have the spinor bundles $\bS_t^\pm := \bS_{M_t}^\pm$, and for each $t\in\R$ we have $\bS_Z|_{M_t} = \bS_t^+\oplus\bS_t^-$. 
From \cref{eq:pos_inner_prod}, we have a positive-definite inner product $\la\cdot|\cdot\ra^\pos_{M_t}$ on $\Gamma_c^\infty(M_t,\bS_Z|_{M_t})$. We consider the Hilbert spaces $\mH_t := L^2(M_t,\bS_Z|_{M_t})$ and $\mH_t^\pm := L^2(M_t,\bS_t^\pm)$, and we note that we have the orthogonal direct sum decomposition $\mH_t = \mH_t^+ \oplus \mH_t^-$. 

For $x\in M$ and $t_0,t_1\in\R$, we can use parallel transport (with respect to the spin connection on $\bS_Z$) along the curve $t\mapsto (x,t)\in Z$ (i.e.\ an integral curve of the vector field $\nu$) to obtain a linear map $\tau_{t_0}^{t_1}\colon(\bS_{t_0}^\pm)_x\to(\bS_{t_1}^\pm)_x$, which is an isometry with respect to the canonical Hermitian structure on $\bS_Z$. 
Consequently, we also obtain a linear map between the Hilbert spaces $\mH_t^\pm := L^2(M_t,\bS_t^\pm)$ of square-integrable spinors on $M_t$. However, to obtain an isometry, we also need to take into account the change of the volume form $\dvol_{M_t}$. Let $\rho_t$ be the unique positive function on $M$ such that
\begin{align}
\label{eq:vol_function}
\dvol_{M_t} &= \rho_t^2 \dvol_{M_0} .
\end{align}
In lack of a better term, we will refer to $\rho_t$ as the \emph{volume function} of $M_t$. 
We note that $\rho_t$ depends not only on $(M_t,g_t)$, but also on the reference volume form of $(M_0,g_0)$. 
In local coordinates, we have $\rho_t = (|g_0|^{-1}|g_t|)^{\frac14}$, where $|g_t| := |\det(g_t)|$. 
We then define the maps $U_t\colon\mH_t\to\mH_0$ by 
\begin{align*}
(U_t\psi)(x) &:= \rho_t \tau_{t}^{0}\psi(x) . 
\end{align*}
The maps $U_t$ are linear (i.e.\ vector space) isomorphisms. Furthermore, each $U_t$ is an isometry with respect to the canonical inner product on $L^2(M_t,\bS_t)$:
\begin{align*}
\la U_t\phi|U_t\psi\ra_{M_0} &= \int_M (U_t\phi|U_t\psi\ra \dvol_{M_0} 
= \int_M ((|g_0|^{-1}|g_t|)^{\frac14} \tau_{t}^{0}\phi|(|g_0|^{-1}|g_t|)^{\frac14} \tau_{t}^{0}\psi) \sqrt{|g_0|} dx^n \\
&= \int_M (\phi|\psi) \sqrt{|g_t|} dx^n 
= \la\phi|\psi\ra_{M_t} .
\end{align*}

Let us write $\gamma_t(\nu) := \gamma_Z(\nu)|_{\bS_t}$. 
In the Lorentzian case, we recall that we have a positive-definite inner product on $L^2(M_t,\bS_t)$ given by the Hermitian structure $(\phi|\psi)^\pos := (\phi|\gamma_t(\nu)\psi)$ on $\bS_t$. 
We check if $U_t$ would also be an isometry with respect to this positive-definite inner product:
\begin{align*}
\la U_t\phi|U_t\psi\ra^\pos_{M_0} &= \int_M (U_t\phi|\gamma_0(\nu)U_t\psi) \dvol_{M_0} 
= \int_M (\phi|\tau_{0}^{t}\gamma_0(\nu)\tau_{t}^{0}\psi) \sqrt{|g_t|} dx^n 
= \la\phi|\gamma_t(\nu)\tau_{0}^{t}\gamma_0(\nu)\tau_{t}^{0}\psi\ra^\pos_{M_t} .
\end{align*}
Hence we see that $U_t$ gives a \emph{unitary} isomorphism $\mH_t\to\mH_0$ if and only if $\gamma_t(\nu) = \tau_{0}^{t}\gamma_0(\nu)\tau_{t}^{0}$, i.e.\ if and only if $\nu$ is \emph{geodesic}.

\begin{assumption}
We assume from now on that the unit normal vector field $\nu$ (which is orthogonal to the hypersurfaces $M_t$) is \emph{geodesic}, i.e.\ $\nabla_\nu\nu=0$. 
\end{assumption}

Locally we can always choose $\nu$ such that it is geodesic, but our assumption that we can do this \emph{globally} places a restriction on the class of space(time)s that we consider. Nevertheless, considering Lorentzian signature, we note that this assumption can be satisfied in interesting physical examples such as Schwarzschild spacetime or the Friedmann-Lema\^itre-Robertson-Walker spacetimes (see e.g.\ \cite[\S2.3]{Poisson04}). 

Since the grading operator on $\bS_Z$ is parallel, we know that the unitary isomorphisms $U_t$ preserve the decomposition $\bS_t = \bS_t^+ \oplus \bS_t^-$ (i.e., $U_t$ is of the form $U_t^+ \oplus U_t^-$, where $U_t^\pm := U_t|_{\bS_t^\pm}$). 
The Clifford multiplication with $\nu$, given by (see \cref{eq:Clifford_rep})
\[
\gamma_Z(\nu) = \mattwo{0}{i\tau_0\epsilon_0}{i\tau_0\epsilon_0}{0} ,
\]
implements the isomorphism $\bS_t^+ \simeq \bS_t^-$. 
Since $\nabla_\nu\nu=0$, we know that $\gamma_Z(\nu)$ commutes with $\nabla^{\bS_Z}_\nu$, so that the unitary isomorphisms $U_t^\pm$ are compatible with the identification $\bS_t^+ \simeq \bS_t^-$ (i.e., $U_t^+ \simeq U_t^-$). 

Consider now the space $C_c(\R,\mH_0)$ of continuous, compactly supported maps from $\R$ to the Hilbert space $\mH_0 := L^2(M_0,\bS_Z|_{M_0})$. Using the canonical inner product on $\mH_0$, we define a (possibly non-degenerate) inner product on $C_c(\R,\mH_0)$ by
\[
\la\phi|\psi\ra := \int_\R \la\phi(t)|\psi(t)\ra_{M_0} dt .
\]
In the Riemannian case, this inner product is positive-definite, and we write $\la\cdot|\cdot\ra^\pos := \la\cdot|\cdot\ra$. In the Lorentzian case, we introduce a positive-definite inner product given by 
\[
\la\phi|\psi\ra^\pos := \int_\R \la\phi(t)|\psi(t)\ra^\pos_{M_0} dt = \int_\R \la\phi(t)|\gamma_0(\nu)\psi(t)\ra_{M_0} dt .
\]
We denote by $L^2(\R,\mH_0)$ the completion of $C_c(\R,\mH_0)$ with respect to $\la\cdot|\cdot\ra^\pos$. We define a map $U\colon\Gamma_c(M\times\R,\bS_Z)\to C_c(\R,\mH_0)$ by 
\begin{align}
\label{eq:unitary}
(U\psi)(t) &:= N_t^{\frac12} \cdot U_t\psi|_{M_t} .
\end{align}
We check that $U$ is an isometry with respect to the canonical inner product:
\begin{align*}
\la U\phi|U\psi\ra &= \int_\R \big\la N_t^{\frac12} \cdot U_t\phi|_{M_t} \bigmvert N_t^{\frac12} \cdot U_t\psi|_{M_t} \big\ra_{M_0} dt 
= \int_\R \int_M N_t \big( U_t\phi|_{M_t} \bigmvert U_t\psi|_{M_t} \big) \dvol_{M_0} dt \\
&= \int_\R \int_M \big( \phi|_{M_t} \bigmvert \psi|_{M_t} \big) N_t \dvol_{M_t} dt 
= \int_{M\times\R} ( \phi | \psi ) \dvol_{M\times\R} 
= \la\phi|\psi\ra .
\end{align*}
In the Lorentzian case, since we assumed $\nu$ to be geodesic, $U$ is also an isometry with respect to the positive-definite inner product.

\subsection{The Dirac operator}
\label{sec:Dirac_decomp}

For each $t\in\R$, we have a Dirac operator $\sD_{M_t}$ on the hypersurface $M_t = M\times\{t\}$. This family of Dirac operators defines an operator $\til\sD_{M_\bullet}$ on $Z = M\times \R$, given by $\big(\til\sD_{M_\bullet} \psi\big)(x,t) := \big(\til\sD_{M_t}\psi|_{M_t}\big)(x)$ for any $\psi\in\Gamma_c^\infty(M\times\R,\bS_Z)$. 
From \cref{eq:Dirac_fol} we then know that the (canonical) Dirac operator $\sD_Z$ on $\Gamma_c^\infty(M\times\R,\bS_Z)$ decomposes as
\[
\sD_Z = \epsilon_0 \gamma_Z(\nu) \nabla^{\bS_Z}_{\nu} - \tau_0 \gamma_Z(\nu) \til\sD_{M_\bullet} - \frac n2 H \gamma_Z(\nu) .
\]
We will express $\sD_Z$ as an operator on $C_c^\infty(\R,\mH_0)$ under the isomorphism $U$. 

The time derivative $\partial_t$ on $C_c^\infty(\R,\mH_0)$ is related to the covariant time derivative $\nabla^{\bS_Z}_{\partial_T}$ on $\Gamma_c^\infty(M\times\R,\bS_Z)$ as follows. (In our notation, we distinguish between the coordinate vector field $\partial_T$ on $Z$ and the differential operator $\partial_t$ on $L^2(\R,\mH_0)$.) 
For $\psi\in\Gamma_c^\infty(M\times\R,\bS_Z)$ we have 
\begin{align*}
\big( \partial_t \circ U \psi \big)(t) &= \lim_{\epsilon\to0} \epsilon^{-1} \big( N_{t+\epsilon}^{\frac12} U_{t+\epsilon}\psi|_{M_{t+\epsilon}} - N_t^{\frac12} U_t\psi|_{M_t} \big) 
= \lim_{\epsilon\to0} \epsilon^{-1} \big( N_{t+\epsilon}^{\frac12} \rho_{t+\epsilon} \tau_{t+\epsilon}^0 \psi|_{M_{t+\epsilon}} - N_t^{\frac12} \rho_t \tau_t^0 \psi|_{M_t} \big) \\
&= \tau_t^0 \Big( \lim_{\epsilon\to0} \epsilon^{-1} \big( \tau_{t+\epsilon}^t N_{t+\epsilon}^{\frac12} \rho_{t+\epsilon} \psi|_{M_{t+\epsilon}} - N_t^{\frac12} \rho_t \psi|_{M_t} \big) \Big) 
= \tau_t^0 \big( (\nabla^{\bS_Z}_{\partial_T} N_t^{\frac12} \rho_t \psi)(t) \big) \\
&= \big( U \circ N_t^{-\frac12} \rho_t^{-1} \circ \nabla^{\bS_Z}_{\partial_T} \circ N_t^{\frac12} \rho_t \psi \big)(t) .
\end{align*}
Hence we have $\partial_t = U\circ N_t^{-\frac12}\rho_t^{-1}\circ\nabla^{\bS_Z}_{\partial_T}\circ N_t^{\frac12}\rho_t\circ U^{-1}$ on $C_c(\R,\mH_0)$. Rewriting this, and using $\nu = N^{-1}\partial_T$, we obtain
\begin{align}
\label{eq:time_derivative}
U \nabla^{\bS_Z}_\nu U^{-1} &= N_t^{-\frac12} \rho_t \partial_t \rho_t^{-1} N_t^{-\frac12} .
\end{align}
We define
\begin{align}
\label{eq:fam_Dirac}
\til\sD_\bullet &:= U \til\sD_{M_\bullet} U^{-1} = U \mattwo{\sD_{M_\bullet}}{0}{0}{-\sD_{M_\bullet}} U^{-1} =: \mattwo{\sD_\bullet}{0}{0}{-\sD_\bullet},
\end{align}
We note that the last equality relies on $\nu$ being geodesic, which ensures that the parallel transports on $\bS_Z^+$ and $\bS_Z^-$ are compatible with their mutual identification via $\gamma_Z(\nu)$. 
We view $\til\sD_\bullet$ as a family of operators $\{\til\sD_t\}_{t\in\R}$ on $\mH_0$. 

By definition, we have $nH := \tr^M(W) = \sum_{j=1}^n g(e_j,W(e_j))$. Using our assumption that $\nabla_\nu\nu = 0$, we therefore find (writing $e_0=\nu$)
\[
\Div \nu = \sum_{j=0}^n g(e_j,\nabla_{e_j}\nu) = - \epsilon_0 \sum_{j=1}^n g(e_j,W(e_j)) = - \epsilon_0 \tr^M(W) = - \epsilon_0 nH .
\]
Furthermore, in terms of local coordinates given by $x^0=T$ and coordinates $x^j$ on $M$, we can also calculate 
\begin{align*}
\Div\nu|_{M_t} &= \sqrt{|g|}^{-1} \partial_t (\sqrt{|g|} N_t^{-1}) = N_t^{-1} \sqrt{|g_t|}^{-1} \partial_t \big(\sqrt{|g_t|}\big) = 2 N_t^{-1} |g_t|^{-\frac14} \partial_t \big(|g_t|^{\frac14}\big) = 2 N_t^{-1} \rho_t^{-1} (\partial_t \rho_t) .
\end{align*}
Combining these equalities we obtain
\begin{align}
\label{eq:rho_H}
- \epsilon_0 \frac n2 H_t &= N_t^{-1} \rho_t^{-1} (\partial_t \rho_t) .
\end{align}
By combining \cref{eq:time_derivative,eq:rho_H}, we see that
\begin{align}
\label{eq:time_derivative2}
U \Big( \epsilon_0 \nabla^{\bS_Z}_\nu - \frac n2 H_t \Big) U^{-1} = \epsilon_0 N_t^{-\frac12} \big( \rho_t \partial_t \rho_t^{-1} + [\partial_t,\rho_t] \rho_t^{-1} \big) N_t^{-\frac12} = \epsilon_0 N_t^{-\frac12} \partial_t N_t^{-\frac12} .
\end{align}
By inserting \cref{eq:fam_Dirac,eq:time_derivative2} into \cref{eq:Dirac_fol}, we obtain:
\begin{prop}
\label{prop:Dirac_decomp}
Let $(Z,g) = (M\times\R,g_\bullet + \epsilon_0 N^2dT^2)$ be a product space(time), such that the unit normal vector field $\nu$ is geodesic. 
Under the isomorphism $U\colon L^2(M\times\R,\bS_Z)\to L^2(\R,\mH_0)$ from \cref{eq:unitary}, 
the canonical Dirac operator $\sD_Z$ on $L^2(Z,\bS_Z)$ is given by 
\begin{align*}
U \sD_Z U^{-1} &= \gamma_0(\nu) \left( \epsilon_0 N_\bullet^{-\frac12} \partial_t N_\bullet^{-\frac12} - \tau_0 \til\sD_\bullet \right) = \mattwo{0}{ i\tau_0 N_\bullet^{-\frac12} \partial_t N_\bullet^{-\frac12} + i \sD_\bullet }{i\tau_0 N_\bullet^{-\frac12} \partial_t N_\bullet^{-\frac12} - i \sD_\bullet }{0} . \nonumber
\end{align*}
\end{prop}
In particular, the Dirac operator $\sD_Z$ can be completely obtained from the families of operators $\{\sD_t\}_{t\in\R}$ and $\{N_t\}_{t\in\R}$. 
This decomposition of the Dirac operator will serve as our motivation for the abstract framework which we develop in the following section. 

\begin{remark}
We note that the volume function has completely disappeared from the above decomposition of the Dirac operator $\sD_Z$. In particular, we do not need the volume function to reconstruct $\sD_Z$. This shouldn't be too surprising; indeed, 
the volume function $\rho_t$ on $M_t$ is not an independent object, but in fact completely determined by the Dirac operator $\sD_t$. An explicit expression can be obtained using the Wodzicki residue density, as follows.
Note that the principal symbol of $\sD_t^2$ is equal to the principal symbol of the Laplacian on $M_t$. From (the proof of) \cite[Proposition 7.7]{GVF01}, we then know that the local Wodzicki residue density $\wres_x\big(|\sD_t|^{-n}\big) d^nx$ is equal to the volume form $\dvol_{g_t} = \sqrt{|g_t|} d^nx$ (up to a constant factor depending only on the dimension $n$). 
The volume function is therefore given by 
\[
\rho_t(x) = \left( \frac{\sqrt{|g_t(x)|}}{\sqrt{|g_0(x)|}} \right)^{\frac12} = \left( \frac{\wres_x\big(|\sD_t|^{-n}\big)}{\wres_x\big(|\sD_0|^{-n}\big)} \right)^{\frac12} .
\]
\end{remark}

\section{Product spectral triples}
\label{sec:prod_spec_trip}

\subsection{Families of operators}

Recall that a family of bounded operators $\{B_t\}_{t\in\R}$ on a Hilbert space $\mH$ is called \emph{strongly continuous} if $B_t\psi$ is (norm-)continuous in $t$ for each $\psi\in\mH$. Similarly, we say that $\{B_t\}_{t\in\R}$ is \emph{weakly continuous} if $\la\xi|B_t\psi\ra$ is continuous in $t$ for all $\xi,\psi\in\mH$. 

We say that $\{B_t\}_{t\in\R}$ is \emph{strongly differentiable} if there exists a strongly continuous family of bounded operators $\{(\partial B)_t\}_{t\in\R}$ such that $\partial_t(B_t\psi) = (\partial B)_t \psi$ for any $\psi\in\mH$. Similarly, we say that $\{B_t\}_{t\in\R}$ is \emph{weakly differentiable} if there exists a weakly continuous family of bounded operators $\{(\partial B)_t\}_{t\in\R}$ such that $\partial_t\big(\la\xi|B_t\psi\ra\big) = \la\xi|(\partial B)_t \psi\ra$ for any $\xi,\psi\in\mH$. 
If no confusion arises, we sometimes write $\partial_tB_t = (\partial B)_t$. 

We point out that all statements below also apply to a family of operators $B_t\colon\mH_1\to\mH_2$ between two different Hilbert spaces, by viewing $B_t$ as an operator on $\mH:=\mH_1\oplus\mH_2$. 

\begin{lem}
\label{lem:weak_cont_loc_bdd}
If $\{B_t\}_{t\in\R}$ is weakly continuous, then it is locally bounded. 
\end{lem}
\begin{proof}
Let $[t_0,t_1]$ be a bounded interval in $\R$. For $\xi,\psi\in\mH$ we know that $\la\xi|B_t\psi\ra$ is continuous, so in particular $\sup_{t\in[t_0,t_1]} \big|\la\xi|B_t\psi\ra\big| < \infty$. The uniform boundedness principle then implies that $\sup_{t\in[t_0,t_1]} \|B_t\| < \infty$.
\end{proof}

Given a strongly continuous family of operators $\{B_t\}_{t\in\R}$ on $\mH$, we define the operator $B_\bullet$ on the Hilbert $C_0(\R)$-module $C_0(\R,\mH)$ by $(B_\bullet\psi)(t) := B_t\psi(t)$. The strong continuity of $\{B_t\}_{t\in\R}$ implies that $B_\bullet$ is well-defined on the initial domain $C_c(\R,\mH)$. If  $\{B_t\}_{t\in\R}$ is \emph{strictly} continuous (i.e.\ the family of adjoints $\{B_t^*\}_{t\in\R}$ is also strongly continuous), then the adjoint $B_\bullet^*$ is also well-defined on $C_c(\R,\mH)$ (so $B_\bullet$ is a \emph{semi-regular} operator). If furthermore $\{B_t\}_{t\in\R}$ is globally bounded (i.e.\ if $\sup_{t\in\R}\|B_t\|<\infty$), then $B_\bullet$ is an adjointable endomorphism on $C_0(\R,\mH)$. 

\begin{lem}
\label{lem:weak_diff_fam}
If $\{B_t\}_{t\in\R}$ is a weakly differentiable family of bounded operators on $\mH$, then the following statements hold.
\begin{enumerate}
\item The family $\{B_t\}_{t\in\R}$ is norm-continuous. 
\item For any $\xi,\psi\in C^1(\R,\mH)$, we have $\la\xi(\cdot)|B_\bullet\psi(\cdot)\ra \in C^1(\R)$. 
\item Let $\{A_t\}_{t\in\R}$ and $\{C_t\}_{t\in\R}$ be strongly differentiable on $\mH$. Then $\{A_t^*B_tC_t\}_{t\in\R}$ is weakly differentiable. 
\end{enumerate}
\end{lem}
\begin{proof}
\begin{enumerate}
\item 
The proof is as in \cite[Remark 8.4, 2.]{KL13}, using the local boundedness of $(\partial B)_t$.

\item For simplicity, suppose that $\xi$ is constant as a function of $t$. Then we calculate
\begin{align*}
&\big\| \big\la\xi \bigmvert \partial_t(B_\bullet\psi)(t) - B_t\partial_t\psi(t) - (\partial B)_t\psi(t) \big\ra \big\| \\
&\qquad= \lim_{\epsilon\to0} \epsilon^{-1} \big\| \big\la\xi \bigmvert (B_\bullet\psi)(t+\epsilon) - (B_\bullet\psi)(t) - B_t(\psi(t+\epsilon)-\psi(t)) - (B_{t+\epsilon} - B_t) \psi(t) \big\ra \big\| \\
&\qquad= \lim_{\epsilon\to0} \epsilon^{-1} \big\| \big\la\xi \bigmvert B_{t+\epsilon}\psi(t+\epsilon) + B_t\psi(t) - B_t\psi(t+\epsilon) - B_{t+\epsilon}\psi(t) \big\ra \big\| \\
&\qquad= \lim_{\epsilon\to0} \epsilon^{-1} \big\| \big\la\xi \bigmvert (B_{t+\epsilon}-B_t) (\psi(t+\epsilon) - \psi(t)) \big\ra \big\| \\
&\qquad\leq \lim_{\epsilon\to0} \|\xi\| \; \big\| (B_{t+\epsilon}-B_t) \big\| \; \epsilon^{-1} \big\| \psi(t+\epsilon) - \psi(t) \big\| \\
&\qquad= \|\xi\| \cdot 0 \cdot \|\partial_t\psi(t)\| = 0 ,
\end{align*}
where on the last line we used the norm-continuity of $B_t$ from the first statement. 
Hence we see that $\la\xi|\partial_t(B_\bullet\psi)(t)\ra = \la\xi|B_t\partial_t\psi(t) + (\partial B)_t\psi(t)\ra$. 
Thus we have proven that $\la\xi|B_t\psi(t)\ra$ is differentiable. The case of $\xi\in C^1(\R,\mH)$ can then be obtained with similar arguments. 

\item For any $\xi,\psi\in\mH$ we have $\la\xi|A_t^*B_tC_t\psi\ra = \la A_t\xi|B_tC_t\psi\ra$. By assumption, $A_t\xi$ and $C_t\psi$ are differentiable, and it then follows from the second statement that $\la\xi|A_t^*B_tC_t\psi\ra$ is differentiable, 
and we have 
\begin{align*}
\partial_t\big(\la\xi|A_t^*B_tC_t\psi\ra\big) &= \big\la\xi \bigmvert \big( (\partial A)_t^* B_t C_t + A_t^* (\partial B)_t C_t + A_t^* B_t (\partial C)_t \big) \psi\big\ra .
\qedhere
\end{align*}
\end{enumerate}
\end{proof}

\begin{lem}
\label{lem:strong_diff_fam}
If $\{B_t\}_{t\in\R}$ is a strongly differentiable family of bounded operators on $\mH$, then the following statements hold.
\begin{enumerate}
\item For any $\psi\in C^1(\R,\mH)$ we have $B_\bullet\psi\in C^1(\R,\mH)$. 
\item The closure of the operator $(\partial B)_\bullet$ on $C_c(\R,\mH)$ equals the closure of the commutator $[\partial_t,B_\bullet]$ on $C_c^1(\R,\mH)$. 
\end{enumerate}
\end{lem}
\begin{proof}
\begin{enumerate}
\item 
The argument is as in the proof of \cref{lem:weak_diff_fam}.2). In this case, for $\psi\in C^1(\R,\mH)$ we have
\begin{align*}
&\big\| \partial_t(B_\bullet\psi)(t) - B_t\partial_t\psi(t) - (\partial B)_t\psi(t) \big\| \leq 
\lim_{\epsilon\to0} \big\| (B_{t+\epsilon}-B_t) \big\| \; \epsilon^{-1} \big\| \psi(t+\epsilon) - \psi(t) \big\| 
= 0 \cdot \|\partial_t\psi(t)\| = 0 .
\end{align*}
Hence $\partial_t(B_\bullet\psi) = B_\bullet\partial_t\psi + (\partial B)_\bullet\psi \in C(\R,\mH)$. 

\item By the proof of the first statement, the commutator $[\partial_t,B_\bullet]$ is well-defined for $\psi\in C_c^1(\R,\mH)$ and given by 
$[\partial_t,B_\bullet]\psi = (\partial B)_\bullet\psi$. 
In particular, the operator $[\partial_t,B_\bullet]$ restricts to a well-defined bounded operator $(\partial B)_t$ on $\mH$ for each $t\in\R$, and these operators are strongly continuous (hence locally bounded). Since $C_c^1(\R,\mH)$ is dense in $C_0(\R,\mH)$, the statement follows. 
\end{enumerate}
\end{proof}

\subsection{Families of spectral triples}

A family of representations $\{\pi_t\}_{t\in\R}$ of a $C^*$-algebra $A$ on a Hilbert space $\mH$ is called \emph{strongly/weakly continuous/differentiable}, if the function $t\mapsto\pi_t(a)\in B(\mH)$ is strongly/weakly continuous/differentiable for every $a\in A$. 
Let $A\odot C_0(\R)$ denote the algebraic tensor product. 
Given a strongly continuous family of representations $\{\pi_t\}_{t\in\R}$, we define for any simple tensor $a\otimes f\in A\odot C_0(\R)$ the operator $\pi_\bullet(a\otimes f)$ on $C_0(\R,\mH)$ by 
\[
\big(\pi_\bullet(a\otimes f) \psi\big)(t) := f(t) \pi_t(a) \psi(t) .
\]
Since any representation of a $C^*$-algebra is norm-decreasing, we have $\|\pi_t\|\leq1$, and therefore $\pi_\bullet$ is bounded. 
Since $A\odot C_0(\R)$ is dense in $C_0(\R,A)$, $\pi_\bullet$ extends to a representation of $C_0(\R,A)$ on $C_0(\R,\mH)$.

\begin{defn}
\label{defn:fam_spec_trip}
A family of spectral triples $\{(\A,{}_{\pi_t}\mH,\D_t)\}_{t\in\R}$ is called \emph{weakly differentiable} if the following conditions are satisfied:
\begin{enumerate}
\item there exists another Hilbert space $W$ which is continuously and densely embedded in $\mH$ such that the inclusion map $\iota\colon W\hookrightarrow\mH$ is \emph{locally compact}, i.e.\ the composition $\pi_t(a)\circ\iota$ is compact for each $t\in\R$ and $a\in\A$; 
\item the domain of $\D_t$ is independent of $t$ and equals $W$, and the graph norm of $\D_t$ is uniformly equivalent to the norm of $W$ (i.e.\ there exist constants $C_1,C_2>0$ such that $C_1\|\xi\|_W \leq \|\xi\|_{\D_t} \leq C_2\|\xi\|_W$ for all $\xi\in W$ and all $t\in\R$);
\item the map $\D_\bullet\colon\R\to\B(W,\mH)$ is weakly differentiable, and its weak derivative is uniformly bounded;
\item the family of representations $\{\pi_t\}_{t\in\R}$ of $A$ on $\mH$ is weakly differentiable, and for each $a\in\A$ the family $\{\pi_t(a)\colon W\to W\}$ is strongly continuous. 
\end{enumerate}
\end{defn}

To avoid confusion, we will sometimes write $\|\cdot\|_W$ for the norm on $W$, and $\|\cdot\|_{W\to\mH}$ for the operator norm on $\B(W,\mH)$. The assumptions on the family of operators $\D_\bullet = \{\D_t\}_{t\in\R}$ are as in \cite[\S8]{KL13}. The above definition is very similar to the families of spectral triples studied in \cite[\S4.3]{vdDR16}, but the assumptions on the family of representations are slightly different here. 

\begin{lem}
\label{lem:comm_s-cts}
Given a weakly differentiable family of spectral triples $\{(\A,{}_{\pi_t}\mH,\D_t)\}_{t\in\R}$, the commutator $[\D_\bullet,\pi_\bullet(a)]$ is strongly continuous for any $a\in\A$.
\end{lem}
\begin{proof}
Using the fact that weak differentiability implies norm-continuity (see \cref{lem:weak_diff_fam}), the statement follows from the inequality
\begin{align*}
\| [\D_t,\pi_t(a)]\psi - [\D_s,\pi_s(a)]\psi \| &\leq \|\D_t-\D_s\|_{W\to\mH} \|\pi_t(a)\psi\|_W + \|\D_s\|_{W\to\mH} \|\pi_t(a)\psi-\pi_s(a)\psi\|_W \\
&\qquad+ \|\pi_t(a)-\pi_s(a)\| \, \|\D_t\psi\| + \|\pi_s(a)\| \, \|\D_t\psi-\D_s\psi\| .
\qedhere
\end{align*}
\end{proof}

\begin{prop}[cf.\ {\cite[Proposition 4.18]{vdDR16}}]
\label{prop:fam_Kasp_mod}
If $\{(\A,{}_{\pi_t}\mH,\D_t)\}_{t\in\R}$ is a weakly 
differentiable family of spectral triples, then the triple $(\A\odot C_c^\infty(\R), C_0(\R,\mH)_{C_0(\R)},\D_\bullet)$ 
is an odd unbounded Kasparov $C_0(\R,A)$-$C_0(\R)$-module. 
\end{prop}
\begin{proof}
The proof that $\D_\bullet$ is self-adjoint and regular is exactly as in \cite[Proposition 4.18]{vdDR16}. The remainder of the proof is only slightly different from the proof of \cite[Proposition 4.18]{vdDR16}, because our assumptions on the representation $\pi_\bullet$ are slightly different. 

The algebraic tensor product $\A\odot C_c^\infty(\R)$ is dense in $C_0(\R,A)$, and for 
$a\otimes f\in \A\odot C_c^\infty(\R)$ the commutators 
$$
\big[ \D_\bullet , \pi_\bullet(a\otimes f) \big](t) = f(t) \big[\D_t,\pi_t(a)\big]  
$$
are bounded for each $t\in\R$. By \cref{lem:comm_s-cts} such commutators are strongly
continuous and therefore locally bounded, and the compact support of $f$ then ensures that they are globally bounded. 

The operator $\pi_t(a) (\D_t\pm i)^{-1}$ is compact and bounded by 
$\|a\|$ for each $t\in\R$ (since $(\A,{}_{\pi_t}\mH,\D_t)$ is a spectral triple). 
The norm-continuity of $\D_t\colon W\to\mH$ implies that the resolvents $(\D_t\pm i)^{-1}$ are also norm-continuous. 
Hence the map $\R\to \mK(\mH)$, $t \mapsto \pi_t(a)(\D_t\pm i)^{-1}$ 
is continuous and globally bounded by $\|a\|$, so if we also multiply by 
$f\in C_0(\R)$ we get $\pi_\bullet(a\otimes f)(\D_\bullet\pm i)^{-1} \in C_0(\R,\mK(\mH))$.
\end{proof}

We consider the balanced tensor product $L^2(\R,\mH) := C_0(\R,\mH) \otimes_{C_0(\R)} L^2(\R)$. 
The operator $\D_\bullet\otimes1$ is well-defined on $\Dom\D_\bullet \otimes_{C_0(\R)} L^2(\R) \subset L^2(\R,\mH)$, and is denoted simply by $\D_\bullet$ as well. 
Furthermore, we consider the operator $\partial_t$ on $L^2(\R,\mH)$. 
Under the isomorphism $L^2(\R,\mH) \simeq \mH \otimes L^2(\R)$, we can identify $\partial_t$ on $L^2(\R,\mH)$ with $1\otimes\partial_t$ on $\mH \otimes L^2(\R)$. 
We note that $C_c^\infty(\R,W)$ is a common core for $\D_\bullet$ and $\partial_t$. 

\begin{thm}[cf.\ {\cite[Theorem 4.20]{vdDR16}}]
\label{thm:spec_prod_triv}
Let $\{(\A,{}_{\pi_t}\mH,\D_t)\}_{t\in\R}$ be a weakly differentiable family of spectral triples. 
Then the operator
$$
\D_\bullet\times(\mp i\partial_t) := \mattwo{0}{\pm i\partial_t + i\D_\bullet}{\pm i\partial_t - i\D_\bullet}{0} \colon \left(\Dom(\D_\bullet)\cap\Dom(\partial_t)\right)^{\oplus2} \to L^2(\R,\mH)^{\oplus2} 
$$
is self-adjoint on the domain $(\Dom\D_\bullet\cap\Dom\partial_t)^{\oplus2}$. 
Furthermore, 
$
\big( \A\odot C_c^\infty(\R) , L^2(\R,\mH)^{\oplus2} , \D_\bullet\times(\mp i\partial_t) \big) 
$
are even spectral triples which represent the odd 
unbounded Kasparov product 
of $\big( \A\odot C_c^\infty(\R), C_0(\R,\mH)_{C_0(\R)},\D_\bullet \big)$ with $\big( C_c^\infty(\R), L^2(\R),\mp i\partial_t \big)$. 
\end{thm}
\begin{proof}
The proof is similar to the proof of \cite[Theorem 4.20]{vdDR16} (which in turn is based on \cite[Proposition 8.11]{KL13}), but we have slightly different assumptions on the representation $\pi_\bullet$. We need to check the boundedness of the commutator $[\partial_t,\pi_\bullet(a\otimes f)\otimes1]$ for all $a\otimes f\in \A\odot C_c^\infty(\R)$. We have
\[
[\partial_t,\pi_\bullet(a\otimes f)\otimes1](t) = \pi_t(a) \partial_t f(t) + f(t) \partial_t \pi_t(a) .
\]
The first term is bounded because $f\in C_c^\infty(\R)$, and the second term is bounded because $\partial_t\pi_t(a)$ is strongly continuous and therefore bounded on the compact support of $f$. 
Thus, as in the proof of \cite[Theorem 4.20]{vdDR16}, this shows that we have a correspondence (as defined in \cite[Definition 6.3]{KL13}) from $(\A\odot C_c^\infty(\R), C_0(\R,\mH)_{C_0(\R)},\D_\bullet)$ to $(C_c^\infty(\R), L^2(\R),\mp i\partial_t)$, and the statement then follows from \cite[Theorems 6.7 \& 7.5]{KL13} (noting that the operator described in \cite{KL13} is unitarily isomorphic to the operator $\D_\bullet\times(\mp i\partial_t)$ defined here). 
\end{proof}

We view a weakly differentiable family of spectral triples as an abstract (noncommutative) analogue of the Dirac operators on a family of (spacelike) hypersurfaces. 
By analogy with the decomposition of the classical Dirac operator described in \cref{sec:Dirac_decomp}, we now also introduce an abstract analogue of lapse functions. 

\begin{defn}
\label{defn:lapse}
Given a weakly differentiable family of spectral triples $\{(\A,\mH,\D_t)\}_{t\in\R}$, we consider a \emph{family of lapse operators} $\{N_t\}_{t\in\R}$ satisfying the following assumptions:
\begin{enumerate}
\item the family $\{N_t\}$ consists of positive invertible operators on $\mH$; 
\item the operators $N_t^{\frac12}$ and their inverses preserve the domain $W$, the family $\{N_t^{\frac12}\colon W\to W\}$ is strongly differentiable and uniformly bounded, and the inverse family $\{N_t^{-\frac12}\colon W\to W\}$ is strongly continuous and uniformly bounded;
\item the strong derivatives $\{(\partial N^{\frac12})_t\}$ and the commutators $\{[\D_t,N_t^{\frac12}]\}$ on $\mH$ are uniformly bounded;
\item 
$[N_t^{\frac12},\pi_t(a)]=0$ for all $a\in\A$. 
\end{enumerate}
\end{defn}

\begin{lem}
\label{lem:sa_conj}
Let $D$ be a self-adjoint operator on a Hilbert space $\mH$. Consider a self-adjoint invertible operator $T\in B(\mH)$ such that $T\cdot\Dom D\subset\Dom D$ and $[D,T]$ is bounded on $\Dom D$. Then:
\begin{enumerate}
\item $DT$ is closed on $\Dom D$;
\item $TDT$ is self-adjoint on $\Dom D$;
\item $T^{-1}DT^{-1}$ is (well-defined and) self-adjoint on $\Dom D$.
\end{enumerate}
\end{lem}
\begin{proof}
\begin{enumerate}
\item From (the proof of) \cite[Lemma 6.2]{Con13} we know that $\Dom D$ is a core for $DT$. Hence, for every $\psi\in\Dom(DT) = \{\psi\in\mH : T\psi\in\Dom\D\}$, we have a sequence $\psi_n\in\Dom D$ such that $\psi_n\to\psi$ and $DT\psi_n\to DT\psi$. Then
\[
D\psi_n = T^{-1} (DT\psi_n - [D,T]\psi_n) \to T^{-1} (DT\psi - [D,T]\psi) ,
\]
and therefore $\psi\in\Dom D$, which proves the first statement. 

\item 
The essential self-adjointness of $TDT$ is proven in \cite[Corollary 6.3]{Con13}. Since $TDT$ is closed by the first statement, this proves the second statement. 

\item 
Using the first statement, we have $T^{-1}\cdot\Dom D = \Dom DT = \Dom D$. Since $[D,T^{-1}] = - T^{-1} [D,T] T^{-1}$ is bounded, the third statement follows from the second statement. 
\qedhere
\end{enumerate}
\end{proof}

\begin{defn}
\label{defn:prod_spec_trip}
Consider a weakly differentiable family of spectral triples $\{(\A,{}_{\pi_t}\mH,\D_t)\}_{t\in\R}$ with a family of lapse operators $\{N_t\}_{t\in\R}$ (as in \cref{defn:fam_spec_trip,defn:lapse}). 
Define the operators $\D_+$ and $\D_-$ on $L^2(\R,\mH)^{\oplus2}$ by 
\begin{align*}
\D_\pm &:= \mattwo{0}{\pm i N_\bullet^{-\frac12} \partial_t N_\bullet^{-\frac12} + i \D_\bullet}{\pm i N_\bullet^{-\frac12} \partial_t N_\bullet^{-\frac12} - i \D_\bullet}{0} . 
\end{align*}
The triples $( \A\odot C_c^\infty(\R) , L^2(\R,\mH)^{\oplus2} , \D_\pm )$ will be referred to as the \emph{product spectral triples} corresponding to the given families. 
\end{defn}
\begin{remark}
We point out that the definition of the operator $\D_+$ corresponds exactly to the reconstruction formula for the Dirac operator on a (Riemannian) product space given in \cref{prop:Dirac_decomp} (note that $\tau_0=1$ in the Riemannian case). The operator $\D_-$ is obtained from $\D_+$ by replacing $\partial_t\to-\partial_t$ (i.e., by reversing the `time'-orientation). 
\end{remark}

\begin{thm}
\label{thm:spec_prod}
Consider the $\Z_2$-grading $\Gamma := 1\oplus(-1)$ on $L^2(\R,\mH)^{\oplus2}$. 
Then the product spectral triples $( \A\odot C_c^\infty(\R) , L^2(\R,\mH)^{\oplus2} , \D_\pm )$ are even spectral triples.
\end{thm}
\begin{proof}
The proof of self-adjointness of $\D_\pm$ proceeds in several steps.
\begin{enumerate}
\item Consider the family of operators
\[
\D_t' := N_t^{\frac12} \D_t N_t^{\frac12} .
\]
We know from \cref{lem:sa_conj} that $\D_t'$ is self-adjoint on the domain $W$. 
It follows from \cref{lem:weak_diff_fam}.3) that $\D_t'$ is weakly differentiable. 
From our assumptions, it follows straightforwardly that the weak derivative of $\D_t'$ is uniformly bounded, 
and that the graph norms of $\D_t'$ are uniformly equivalent. 
Hence $\{(\A,\mH,\D_t')\}_{t\in\R}$ is again a weakly differentiable family of spectral triples. 

\item By \cref{thm:spec_prod_triv}, the operators 
\[
\D_\pm' := \D_\bullet' \times (\mp i\partial_t) = \mattwo{0}{\pm i\partial_t + i\D_\bullet'}{\pm i\partial_t - i\D_\bullet'}{0}
\]
are self-adjoint on the domain $\big(\Dom\D_\bullet\cap\Dom\partial_t\big)^{\oplus2}$. 

\item Finally, since we have the equality
\begin{align*}
N_\bullet^{-\frac12} ( \pm i \partial_t - i \D_\bullet' ) N_\bullet^{-\frac12} &= \pm i N_\bullet^{-\frac12} \partial_t N_\bullet^{-\frac12} - i \D_\bullet ,
\end{align*}
we see that 
\[
\D_\pm = N_\bullet^{-\frac12} \D_\pm' N_\bullet^{-\frac12} . 
\]
The operator $N_\bullet^{\frac12}$ is a bounded and invertible operator which preserves the domain of $\D_\pm'$, such that the commutator $[\D_\pm',N_\bullet^{\frac12}]$ is bounded. 
It then follows from \cref{lem:sa_conj} that $\D_\pm$ is self-adjoint on $\big(\Dom\D_\bullet\cap\Dom\partial_t\big)^{\oplus2}$. 
\end{enumerate}

Since $N_\bullet$ commutes with the action of $\A\odot C_c^\infty(\R)$, it follows as in the proof of \cref{thm:spec_prod_triv} that $[\D_\pm,\pi_\bullet(a\otimes f)]$ is bounded for any $a\otimes f\in\A\odot C_c^\infty(\R)$. 
Furthermore, we know from \cref{thm:spec_prod_triv} that $\D_\pm'$ has locally compact resolvents. Since $\Dom\D_\pm' = \Dom\D_\pm$, it follows that also $\D_\pm$ has locally compact resolvents. 
Finally, we obviously have $\D_\pm\Gamma = -\Gamma\D_\pm$. 
\end{proof}

\begin{thm}
\label{thm:spec_prod_Kasp}
The product spectral triples $( \A\odot C_c^\infty(\R) , L^2(\R,\mH)^{\oplus2} , \D_\pm )$ represent the odd (internal) 
unbounded 
Kasparov product (over $C_0(\R)$) of $\big( \A\odot C_c^\infty(\R), C_0(\R,\mH)_{C_0(\R)},\D_\bullet \big)$ with $\big( C_c^\infty(\R), L^2(\R),\mp i\partial_t \big)$. 
\end{thm}
\begin{proof}
In view of \cref{thm:spec_prod_triv}, we only need to prove that $\D_\pm$ is homotopic to $\D_\bullet\times(\mp i\partial_t)$. We consider the unbounded Kasparov $A\otimes C_0(\R)$-$C([0,1])$-module $\big( \A\odot C_c^\infty(\R) , C\big([0,1],L^2(\R,\mH)\big)^{\oplus2} , \D_\pm(\cdot) \big)$, where $\D_\pm(\cdot) = \{\D_\pm(s)\}_{s\in[0,1]}$, and $\D_\pm(s)$ is the operator constructed (as in \cref{defn:prod_spec_trip}) from the families $\{(\A,{}_{\pi_t}\mH,\D_t)\}_{t\in\R}$ and $\{N_t(s)\}_{t\in\R}$, where
\begin{align*}
N_t(s)^{\frac12} &:= s + (1-s) N_t^{\frac12} .
\end{align*}
In other words, we obtain the homotopy $\D_\pm(s)$ by connecting $N_t^{\frac12}$ to the identity via a straight line. We note that the operators $N_t(s)$ again satisfy \cref{defn:lapse}. 
To show that $\D_\pm(\cdot)$ indeed defines an unbounded Kasparov module, it suffices to check that the resolvents of $\D_\pm(s)$ are norm-continuous. 

The family $\{N_\bullet(s)^{\frac12}\}_{s\in[0,1]}$ of bounded positive invertible operators on $L^2(\R,W)$ depends continuously on $s$. Consequently, $\D_\pm(s)\colon\Dom\D_\pm\to L^2(\R,\mH)^{\oplus2}$ is norm-continuous (where we view $\Dom\D_\pm$ as a Hilbert space equipped with the graph norm). From \cite[Lemma 3.1]{vdD17apre} it then follows that the resolvents of $\D_\pm(s)$ are indeed norm-continuous. Hence $\big( \A\odot C_c^\infty(\R) , C\big([0,1],L^2(\R,\mH)\big)^{\oplus2} , \D_\pm(\cdot) \big)$ is an unbounded Kasparov $A\otimes C_0(\R)$-$C([0,1])$-module. The bounded transform of $\D_\pm(\cdot)$ then yields a homotopy between the bounded transforms of $\D_\pm$ and $\D_\bullet\times(\mp i\partial_t)$. 
\end{proof}

\begin{prop}
\label{prop:class_fams}
Let $(Z,g) = (M\times\R,g_\bullet + \epsilon_0 N^2dT^2)$ be a product space(time), such that the unit normal vector field $\nu$ is geodesic, and the metrics $g_t$ are complete. 
Suppose that the metrics $g_t$ and the lapse function $N$ have derivatives of all orders (both in $t$ and along $M$) which are \emph{globally bounded}. 
Assume furthermore that $N_t$ is uniformly invertible. 
Then the operators $\sD_t$ from \cref{eq:fam_Dirac} yield a weakly differentiable family of spectral triples $\{( C_c^\infty(M) , \mH_0^+ , \sD_t )\}_{t\in\R}$, and the family $\{N_t\}_{t\in\R}$ satisfies the conditions in \cref{defn:lapse}. 
\end{prop}
\begin{proof}
We know from \cref{prop:ST_t} that $\big( C_c^\infty(M) , \mH_0^+ , \sD_t \big)$ is a spectral triple for each $t\in\R$. 
We note that the representation of $C_c^\infty(M)$ on $\mH_0^+$ is independent of $t$ (since it commutes with parallel transport), so that the continuity and differentiability conditions for the family of representations are automatically satisfied. 
The conditions on the family $\{\sD_t\}$ are proven as in \cite[Proposition 4.22]{vdDR16}. 
Furthermore, the conditions on $\{N_t\}$ are satisfied, because $N_t$ is smooth with uniformly bounded derivatives. 
\end{proof}

Let $(Z,g) = (M\times\R,g_\bullet + N^2dT^2)$ be a product space satisfying the assumptions in \cref{prop:class_fams}, giving a weakly differentiable family of spectral triples $\{( C_c^\infty(M) , \mH_0^+ , \sD_t )\}_{t\in\R}$ and a family of lapse operators $\{N_t\}_{t\in\R}$. 
\cref{thm:spec_prod} then yields the corresponding \emph{product spectral triple} $( \A\odot C_c^\infty(\R) , L^2(\R,\mH)^{\oplus2} , \D_+ )$. Since the definition of $\D_+$ is based on the formula from \cref{prop:Dirac_decomp}, 
we find that the product spectral triple corresponding to a product space correctly reconstructs the canonical Dirac operator on this space. 
\begin{prop}
Let $(Z,g) = (M\times\R,g_\bullet + N^2dT^2)$ be a product space satisfying the assumptions in \cref{prop:class_fams}. Then the corresponding product spectral triple $( C_c^\infty(M)\odot C_c^\infty(\R) , L^2(\R,\mH_0^+)^{\oplus2} , \D_+ )$ is unitarily equivalent to the canonical spectral triple $( C_c^\infty(Z) , L^2(Z,\bS_Z) , \sD_Z )$ on the space $Z$. 
\end{prop}
In particular, the product spectral triple corresponding to a product space is independent (up to unitary equivalence) of the choice of the splitting $Z \simeq M\times\R$.

\section{Lorentzian product triples}
\label{sec:Lor_prod}

\begin{defn}
\label{defn:Lor_prod}
Consider a weakly differentiable family of spectral triples $\{(\A,{}_{\pi_t}\mH,\D_t)\}_{t\in\R}$ with a family of lapse operators $\{N_t\}_{t\in\R}$ (as in \cref{defn:fam_spec_trip,defn:lapse}). 
We define the even \emph{Lorentzian product triple} as the triple $( \A\odot C_c^\infty(\R) , L^2(\R,\mH)^{\oplus2} , \D )$ with
\begin{align*}
\D &:= \mJ \big( - N_\bullet^{-\frac12} \partial_t N_\bullet^{-\frac12} - i \til\D_\bullet \big) = \mattwo{0}{-N_\bullet^{-\frac12} \partial_t N_\bullet^{-\frac12} + i \D_\bullet }{-N_\bullet^{-\frac12} \partial_t N_\bullet^{-\frac12} - i \D_\bullet }{0} ,
\end{align*}
where we have written
\begin{align*}
\til\D_\bullet &:= \mattwo{\D_\bullet}{0}{0}{-\D_\bullet} , & 
\mJ &:= \mattwo{0}{1}{1}{0} .
\end{align*}
\end{defn}

\begin{remark}
We observe that a Lorentzian product triple is equal to the `reverse Wick rotation' (using the terminology of \cite{vdDR16}) of the product spectral triples $( \A \odot C_c^\infty(\R) , L^2(\R,\mH)^{\oplus2} , \D_\pm )$ from \cref{thm:spec_prod}. Indeed, we have the equality
$$
\D = \frac12(\D_++\D_-) + \frac i2(\D_+-\D_-) . 
$$
\end{remark}

Let $(Z,g) = (M\times\R,g_\bullet - N^2dT^2)$ be a product spacetime satisfying the assumptions in \cref{prop:class_fams}, giving a weakly differentiable family of spectral triples $\{( C_c^\infty(M) , \mH_0^+ , \sD_t )\}_{t\in\R}$ and a family of lapse operators $\{N_t\}_{t\in\R}$. 
We note that the assumptions of \cref{prop:class_fams} almost imply that the spacetime is globally hyperbolic: whenever $g_t$ is uniformly bounded below by a complete metric, it follows from \cref{thm:product_glob_hyp} that $(Z,g)$ is globally hyperbolic. 

From \cref{defn:Lor_prod} we obtain the \emph{Lorentzian product triple} $( C_c^\infty(M)\odot C_c^\infty(\R) , L^2(\R,\mH_0^+)^{\oplus2} , \D )$. Since the definition of $\D$ is based on the formula from \cref{prop:Dirac_decomp}, 
we find that the Lorentzian product triple corresponding to a product spacetime correctly reconstructs the canonical Dirac operator on this spacetime. 
\begin{prop}
Let $(Z,g) = (M\times\R,g_\bullet - N^2dT^2)$ be a product spacetime satisfying the assumptions in \cref{prop:class_fams}. Then the corresponding Lorentzian product triple $( C_c^\infty(M)\odot C_c^\infty(\R) , L^2(\R,\mH_0^+)^{\oplus2} , \D )$ is unitarily equivalent to the canonical spectral triple $( C_c^\infty(Z) , L^2(Z,\bS_Z) , \sD_Z )$ on the spacetime $Z$. 
\end{prop}
In particular, the Lorentzian product triple corresponding to a product spacetime is independent (up to unitary equivalence) of the choice of the splitting $Z \simeq M\times\R$.

\subsection{Krein spaces}
\label{sec:Krein}

In the remainder of this article, we want to compare our construction of Lorentzian product triples with other approaches to `Lorentzian spectral triples' that have appeared in the literature. We will focus on the Krein space approach, which has been studied in e.g.\ \cite{Str06,Sui04,PS06,Bar07,vdD16,Bes16pre}. First, we will recall some facts about Krein spaces. 
In the next subsection, we will show that our Lorentzian product triples satisfy the definition of Lorentz-type spectral triples from \cite{vdD16}.

The following summary is based on \cite[\S2]{vdD16}. For a detailed introduction to Krein spaces, we refer to \cite{Bognar74}. 
A Krein space is a vector space $\mH$ with a non-degenerate inner product $\la\cdot|\cdot\ra$ which admits a fundamental decomposition $\mH = \mH^+ \oplus \mH^-$ (i.e., an orthogonal direct sum decomposition into a positive-definite subspace $\mH^+$ and a negative-definite subspace $\mH^-$) such that $\mH^+$ and $\mH^-$ are intrinsically complete (i.e., complete with respect to the norms $\Vert \psi\Vert_{\mH^\pm} := |\la\psi|\psi\ra|^{1/2}$).

A fundamental symmetry $\mJ$ is a self-adjoint unitary operator $\mJ\colon\mH\to\mH$ such that $(1+\mJ)\mH$ is positive-definite and $(1-\mJ)\mH$ is negative-definite. Given a fundamental decomposition $\mH = \mH^+ \oplus \mH^-$, we obtain a corresponding fundamental symmetry $\mJ = P^+ - P^-$, where $P^\pm$ denotes the projection onto $\mH^\pm$. 
Given a fundamental symmetry $\mJ$, we denote by $\mH_\mJ$ the corresponding Hilbert space for the positive-definite inner product $\la\cdot|\cdot\ra_\mJ := \la\mJ\cdot|\cdot\ra$. 

For an operator $T$, we will denote by $T^+$ the \emph{Krein-adjoint} (i.e., the adjoint operator with respect to the Krein inner product $\la\cdot|\cdot\ra$). By the \emph{adjoint} $T^*$ we will mean the usual adjoint in the Hilbert space $\mH_\mJ$ (i.e., with respect to the positive-definite inner product $\la\cdot|\cdot\ra_\mJ$). These adjoints are related via $T^+ = \mJ T^* \mJ$. 

A Krein space $\mH$ with fundamental symmetry $\mJ$ is called \emph{$\Z_2$-graded} if $\mH_\mJ$ is $\Z_2$-graded and $\mJ$ is homogeneous. 
The assumption that $\mH_\mJ$ is $\Z_2$-graded means we have a decomposition $\mH^0\oplus\mH^1$, and that this decomposition is respected by the positive-definite inner product $\la\cdot|\cdot\ra_\mJ$ (which means that $\la\psi_0|\psi_1\ra_\mJ = 0$ for all $\psi_0\in\mH^0$ and $\psi_1\in\mH^1$). The bounded operators $\B(\mH)$ then also decompose into a direct sum of even operators $\B^0(\mH)$ and odd operators $\B^1(\mH)$. The assumption that the fundamental symmetry $\mJ$ is homogeneous means that $\mJ$ is either even or odd. If $\mJ$ is odd, it implements a unitary isomorphism $\mH^0\simeq\mH^1$. Given the decomposition $\mH^0\oplus\mH^1$, we have a (self-adjoint, unitary) grading operator $\Gamma$ which acts as $(-1)^j$ on $\mH^j$ (for $j\in\Z_2$). If $\mJ$ is odd, we note that $\Gamma$ is Krein-\emph{anti}-self-adjoint (indeed, $\Gamma^+ = \mJ\Gamma\mJ = -\Gamma\mJ^2 = -\Gamma$). 

As in \cite[\S2.1]{vdDR16} we define the `combined graph inner product' $\la\cdot|\cdot\ra_{S,T}$ of two closed operators $S$ and $T$ as $\la\psi|\phi\ra_{S,T} := \la\psi|\phi\ra_\mJ + \la S\psi|S\phi\ra_\mJ + \la T\psi|T\phi\ra_\mJ$ (using the \emph{positive-definite} inner product $\la\cdot|\cdot\ra_\mJ$), for all $\psi,\phi\in\Dom S\cap\Dom T$. This inner product yields the corresponding `combined graph norm' $\|\cdot\|_{S,T}$. 
For a Krein-self-adjoint operator $\D$ we have $\mJ\D^* = \D\mJ$ and $\Dom\D^* = \Dom\D\mJ = \mJ\cdot\Dom\D$. 
One can then check that $\la\cdot|\cdot\ra_{\D,\D^*}$ is identical to $\la\cdot|\cdot\ra_{\D\mJ,\mJ\D}$ on $\Dom\D\cap\Dom\D^* = \Dom\D\cap\mJ\cdot\Dom\D$.

\subsection{Lorentz-type spectral triples}

\begin{defn}[{\cite[Definition 2.2]{vdD16}}]
\label{defn:Krein_triple}
A \emph{Lorentz-type spectral triple} $(\A,\mH,\D,\mJ)$ consists of 
\begin{itemize}
\item a $\Z_2$-graded Krein space $\mH$; 
\item a trivially graded $*$-algebra $\A$ along with an even $*$-algebra representation $\pi\colon\A\to B^0(\mH)$;  
\item an odd fundamental symmetry $\mJ$ which commutes with the algebra $\A$; 
\item a densely defined, closed, odd operator $\D\colon\Dom\D\to\mH$ such that:
\begin{enumerate}
\item the linear subspace $\E:=\Dom\D\cap\mJ\cdot\Dom\D$ is dense in $\mH$;
\item the operator $\D$ is Krein-self-adjoint on $\E$;
\item we have the inclusion $\pi(\A)\cdot\E\subset\E$, and the commutator $[\D,\pi(a)]$ is bounded on $\E$ for each $a\in\A$;
\item the map $\pi(a)\circ\iota\colon \E\into \mH\to \mH$ is compact for each $a\in \A$, where $\iota$ denotes the natural inclusion map $\E\into \mH$, and $\E$ is considered as a Hilbert space with the inner product $\la\cdot|\cdot\ra_{\D\mJ,\mJ\D}$.
\end{enumerate}
\end{itemize}
\end{defn}

\begin{prop}
\label{prop:Krein_triple_triv}
Let $\{(\A,{}_{\pi_t}\mH,\D_t)\}_{t\in\R}$ be a weakly differentiable family of spectral triples. 
Then the operator 
\begin{align*}
\D &:= \mattwo{0}{-\partial_t+i\D_\bullet}{-\partial_t-i\D_\bullet}{0} = \mJ \big( -\partial_t - i \til\D_\bullet \big) , & 
\mJ &= \mattwo{0}{1}{1}{0} 
\end{align*}
defines a Lorentz-type spectral triple $( \A\odot C_c^\infty(\R) , L^2(\R,\mH)^{\oplus2} , i\D , \mJ)$. 
\end{prop}
\begin{proof}
We view the Hilbert space $L^2(\R,\mH)^{\oplus2}$ as a Krein space with the fundamental symmetry $\mJ$, and we note that $\mJ\Gamma=-\Gamma\mJ$. 
By assumption, the operator $[\partial_t,\D_\bullet](\D_\bullet\pm i)^{-1}$ is well-defined and bounded. By \cite[Proposition 2.13]{vdDR16}, this implies that $-i\partial_t\pm\D_\bullet$ is essentially self-adjoint on the intersection $\Dom\D_\bullet\cap\Dom\partial_t$. Hence the Krein-adjoint of $\D$ is given on $(\Dom\D_\bullet\cap\Dom\partial_t)^{\oplus2}$ by 
\begin{align*}
\D^+ &= \mJ\D^*\mJ = \mattwo{0}{1}{1}{0} \mattwo{0}{(-\partial_t-i\D_\bullet)^*}{(-\partial_t+i\D_\bullet)^*}{0} \mattwo{0}{1}{1}{0} \\
&= \mattwo{0}{(-\partial_t+i\D_\bullet)^*}{(-\partial_t-i\D_\bullet)^*}{0} = \mattwo{0}{\partial_t-i\D_\bullet}{\partial_t+i\D_\bullet}{0} = -\D ,
\end{align*}
which shows that $i\D$ is essentially Krein-self-adjoint on $\big(\Dom\D_\bullet\cap\Dom\partial_t\big)^{\oplus2}$. The intersection $\Dom\D\cap\Dom\D^*$ contains $\big(\Dom\D_\bullet\cap\Dom\partial_t\big)^{\oplus2}$. 
Furthermore, we have the following equalities (which hold on $(\Dom\D_\bullet\cap\Dom\partial_t)^{\oplus2}$)
\begin{align*}
\frac12(\D+\D^*) &= \mattwo{0}{i\D_\bullet}{-i\D_\bullet}{0} , & 
\frac i2(\D-\D^*) &= \mattwo{0}{-i\partial_t}{-i\partial_t}{0} . 
\end{align*}
Since $\Dom\D_\bullet\cap\Dom\partial_t$ is a core for $\D_\bullet$, the first equality shows that $\frac12(\D+\D^*)$ is a symmetric extension of an essentially self-adjoint operator, and therefore the domain $\Dom\D\cap\Dom\D^*$ must be contained in $(\Dom\D_\bullet)^{\oplus2}$. Similarly, since $\Dom\D_\bullet\cap\Dom\partial_t$ is also a core for $\partial_t$, the second equality implies that $\Dom\D\cap\Dom\D^*$ must be contained in $(\Dom\partial_t)^{\oplus2}$. 
Hence we have the equality $\Dom\D\cap\Dom\D^* = \big(\Dom\D_\bullet\cap\Dom\partial_t\big)^{\oplus2}$. By \cref{thm:spec_prod_triv}, it follows that $\Dom\D\cap\Dom\D^* = \Dom\D_\pm$. 
We consider $\Dom\D\cap\Dom\D^*$ equipped with the combined graph norm $\|\cdot\|_{\D,\D^*}$, and $\Dom\D_\pm$ equipped with the graph norm of $\D_\pm$. 
Since $\|\cdot\|_{\D,\D^*} = \|\cdot\|_{\D_+,\D_-}$ by \cite[Lemma 2.3]{vdDR16}, the identity map $\Dom\D\cap\Dom\D^* \to \Dom\D_\pm$ is bounded. 
Since $\pi_\bullet(a\otimes f)\circ\iota \colon \Dom\D_\pm \to L^2(\R,\mH)^{\oplus2}$ is compact for each $a\otimes f\in \A\odot C_c^\infty(\R)$, it follows that also $\pi_\bullet(a\otimes f)\circ\iota \colon \Dom\D\cap\Dom\D^* \to L^2(\R,\mH)^{\oplus2}$ is compact. 
Finally, the boundedness of the commutator $[\D,\pi_\bullet(a\otimes f)]$ follows as in the proof of \cref{thm:spec_prod_triv}. 
\end{proof}

\begin{thm}
\label{thm:Krein_triple}
Consider a weakly differentiable family of spectral triples $\{(\A,{}_{\pi_t}\mH,\D_t)\}_{t\in\R}$ with a family of lapse operators $\{N_t\}_{t\in\R}$ (as in \cref{defn:fam_spec_trip,defn:lapse}). 
Then the operator $\D$ defined in \cref{defn:Lor_prod} defines a Lorentz-type spectral triple $( \A\odot C_c^\infty(\R) , L^2(\R,\mH)^{\oplus2} , i\D , \mJ)$. 
\end{thm}
\begin{proof}
The idea of the proof is the same as in \cref{thm:spec_prod}, and we leave out some details. Here, we reduce the problem to the special case of \cref{prop:Krein_triple_triv}. 
First, as in the proof of \cref{thm:spec_prod}, we consider the weakly differentiable family of spectral triples $\{(\A,\mH,\D_t')\}_{t\in\R}$. 
Second, by \cref{prop:Krein_triple_triv}, the operator
\[
i\D' := i\mJ ( - \partial_t - i \til\D_\bullet' )
\]
is Krein-self-adjoint. Equivalently, this means that $i\mJ\D'$ is self-adjoint. 
Third, we have the equality $i\D = iN_\bullet^{-\frac12} \D' N_\bullet^{-\frac12}$. 
By \cref{lem:sa_conj}, using that $\mJ$ commutes with $N_\bullet$, the operator $i\mJ\D$ is self-adjoint. Thus we have shown that $i\D$ is Krein-self-adjoint. 
Furthermore, since $N_t$ commutes with $\pi_t(\A)$, the commutators $[\D,\pi_\bullet(a\otimes f)]$ are again bounded for all $a\otimes f\in \A\odot C_c^\infty(\R)$. 
Finally, since $\Dom\D=\Dom\D'$, it also follows from \cref{prop:Krein_triple_triv} that $\pi_\bullet(a\otimes f)\circ\iota\colon\Dom\D\cap\Dom\D^*\to L^2(\R,\mH)^{\oplus2}$ is compact for any $a\otimes f\in\A\odot C_c^\infty(\R)$. 
\end{proof}


\begin{thebibliography}{{\noopsort{Dungen}}DPR13}

\bibitem[AG14]{AG14}
J.~Aastrup and J.~M. Grimstrup, \emph{From quantum gravity to quantum field
  theory via noncommutative geometry}, Class. Quant. Grav. \textbf{31} (2014),
  no.~3, 035018.

\bibitem[Bar07]{Bar07}
J.~W. Barrett, \emph{{A Lorentzian version of the non-commutative geometry of
  the standard model of particle physics}}, J. Math. Phys. \textbf{48} (2007),
  012303.

\bibitem[Bau81]{Baum81}
H.~Baum, \emph{{Spin-Strukturen und Dirac-Operatoren \"uber
  pseudo-Rie\-mann\-schen Mannigfaltigkeiten}}, {Teubner-Texte zur Mathematik},
  vol.~41, Teub\-ner-Verlag, Leipzig, 1981.

\bibitem[BB17]{BB17}
N.~Bizi and F.~Besnard, \emph{The disappearance of causality at small scale in
  almost-commutative manifolds}, J. Math. Phys. \textbf{58} (2017), no.~9,
  092301.

\bibitem[BEE96]{BEE96}
J.~Beem, P.~Ehrlich, and K.~Easley, \emph{{Global Lorentzian Geometry}}, 2 ed.,
  Chapman \& Hall/CRC Pure and Applied Mathematics, Taylor \& Francis, 1996.

\bibitem[Bes09]{Bes09}
F.~Besnard, \emph{A noncommutative view on topology and order}, J. Geom. Phys.
  \textbf{59} (2009), no.~7, 861--875.

\bibitem[Bes16]{Bes16pre}
\bysame, \emph{{On the definition of spacetimes in Noncommutative Geometry,
  part II}}, 2016, arXiv:1611.07842.

\bibitem[BGM05]{BGM05}
C.~B{\"a}r, P.~Gauduchon, and A.~Moroianu, \emph{Generalized cylinders in
  semi-{R}iemannian and spin geometry}, Math. Z. \textbf{249} (2005), no.~3,
  545--580.

\bibitem[Bog74]{Bognar74}
J.~Bogn\'ar, \emph{Indefinite inner product spaces}, Ergebnisse Mathematik und
  GrenzGeb., Springer, Berlin, 1974.

\bibitem[BS05]{BS05}
A.~N. Bernal and M.~S{\'a}nchez, \emph{{Smoothness of time functions and the
  metric splitting of globally hyperbolic space-times}}, Commun. Math. Phys.
  \textbf{257} (2005), 43--50.

\bibitem[CC02]{CC02}
Y.~{Choquet-Bruhat} and S.~Cotsakis, \emph{Global hyperbolicity and
  completeness}, J. Geom. Phys. \textbf{43} (2002), no.~4, 345--350.

\bibitem[Con94]{Connes94}
A.~Connes, \emph{{Noncommutative Geometry}}, Academic Press, San Diego, CA,
  1994.

\bibitem[Con13]{Con13}
\bysame, \emph{{On the spectral characterization of manifolds}}, J. Noncommut.
  Geom. \textbf{7} (2013), 1--82.

\bibitem[DP86]{DP86}
L.~D{\k{a}}browski and R.~Percacci, \emph{Spinors and diffeomorphisms}, Commun.
  Math. Phys. \textbf{106} (1986), no.~4, 691--704.

\bibitem[Dun16]{vdD16}
K.~{\noopsort{Dungen}}{van den}~{Dungen}, \emph{Krein spectral triples and the
  fermionic action}, Math. Phys. Anal. Geom. \textbf{19} (2016), 4.

\bibitem[Dun17]{vdD17apre}
\bysame, \emph{{The index of generalised Dirac-Schr\"odinger operators}}, 2017,
  arXiv:1710.09206.

\bibitem[{\noopsort{Dungen}}DPR13]{vdDPR13}
K.~{\noopsort{Dungen}}{van den}~Dungen, M.~{Paschke}, and A.~{Rennie},
  \emph{{Pseudo-Riemannian spectral triples and the harmonic oscillator}}, J.
  Geom. Phys. \textbf{73} (2013), 37--55.

\bibitem[{\noopsort{Dungen}}DR16]{vdDR16}
K.~{\noopsort{Dungen}}{van den}~Dungen and A.~{Rennie}, \emph{Indefinite
  {K}asparov modules and pseudo-{R}iemannian manifolds}, Ann. Henri
  Poincar{\'e} \textbf{17} (2016), no.~11, 3255--3286.

\bibitem[FE13]{FE13}
N.~{Franco} and M.~{Eckstein}, \emph{{An algebraic formulation of causality for
  noncommutative geometry}}, Class. Quant. Grav. \textbf{30} (2013), no.~13,
  135007.

\bibitem[FE14]{FE14}
\bysame, \emph{Exploring the causal structures of almost commutative
  geometries}, SIGMA \textbf{10} (2014), 10.

\bibitem[FE15]{FE15}
\bysame, \emph{Causality in noncommutative two-sheeted space-times}, J. Geom.
  Phys. \textbf{96} (2015), 42--58.

\bibitem[Fra10]{Fra10}
N.~Franco, \emph{Global eikonal condition for {L}orentzian distance function in
  noncommutative geometry}, SIGMA \textbf{6} (2010), 64.

\bibitem[Fra14]{Fra14}
\bysame, \emph{Temporal {L}orentzian spectral triples}, Rev. Math. Phys.
  \textbf{26} (2014), no.~08, 1430007.

\bibitem[Fra18]{Fra18}
\bysame, \emph{The {L}orentzian distance formula in noncommutative geometry},
  J. Phys.: Conf. Ser. \textbf{968} (2018), no.~1, 012005.

\bibitem[GK99]{GK99}
E.~{Garc{\'\i}a-R{\'\i}o} and D.~Kupeli, \emph{Semi-{R}iemannian maps and their
  applications}, Mathematics and Its Applications, vol. 475, Springer
  Netherlands, 1999.

\bibitem[GVF01]{GVF01}
J.~{Gracia-Bond{\'i}a}, J.~V{\'a}rilly, and H.~Figueroa, \emph{{Elements of
  Noncommutative Geometry}}, Birkh{\"a}user Advanced Texts, 2001.

\bibitem[Haw97]{Haw97}
E.~Hawkins, \emph{{Hamiltonian gravity and noncommutative geometry}}, Commun.
  Math. Phys. \textbf{187} (1997), 471--489.

\bibitem[HR00]{Higson-Roe00}
N.~Higson and J.~Roe, \emph{{Analytic K-Homology}}, Oxford University Press,
  New York, 2000.

\bibitem[KL13]{KL13}
J.~Kaad and M.~Lesch, \emph{Spectral flow and the unbounded {K}asparov
  product}, Adv. Math. \textbf{248} (2013), 495--530. \MR{3107519}

\bibitem[Kop98]{Kop98}
T.~Kopf, \emph{Spectral geometry and causality}, Int. J. Modern Phys. A
  \textbf{13} (1998), no.~15, 2693--2708.

\bibitem[KP01]{KP01}
T.~Kopf and M.~Paschke, \emph{Spectral quadruples}, Modern Phys. Lett. A
  \textbf{16} (2001), no.~04n06, 291--298.

\bibitem[KP02]{KP02}
\bysame, \emph{A spectral quadruple for de {S}itter space}, J. Math. Phys.
  \textbf{43} (2002), no.~2, 818--846.

\bibitem[LM89]{Lawson-Michelsohn89}
H.~Lawson and M.~Michelsohn, \emph{{Spin Geometry}}, Princeton mathematical
  series, Princeton University Press, 1989.

\bibitem[Mil65]{Mil65}
J.~W. Milnor, \emph{Remarks concerning spin manifolds}, Differential and
  Combinatorial Topology: A Symposium in Honor of Marston Morse (S.~S. Cairns,
  ed.), Princeton University Press, 1965, pp.~55--62.

\bibitem[{Min}17]{Min17pre}
E.~{Minguzzi}, \emph{{Causality theory for closed cone structures with
  applications}}, 2017, arXiv:1709.06494.

\bibitem[Mor03]{Mor03}
V.~Moretti, \emph{{Aspects of noncommutative Lorentzian geometry for globally
  hyperbolic space-times}}, Rev. Math. Phys. \textbf{15} (2003), 1171--1217.

\bibitem[Poi04]{Poisson04}
E.~Poisson, \emph{A relativist's toolkit: The mathematics of black-hole
  mechanics}, Cambridge University Press, 2004.

\bibitem[PS06]{PS06}
M.~Paschke and A.~Sitarz, \emph{{Equivariant Lorentzian Spectral Triples}},
  2006, arXiv:math-ph/0611029.

\bibitem[PV04]{PV04}
M.~{Paschke} and R.~{Verch}, \emph{{Local covariant quantum field theory over
  spectral geometries}}, Class. Quant. Grav. \textbf{21} (2004), 5299--5316.

\bibitem[RW16]{RW16}
A.~Rennie and B.~E. Whale, \emph{Generalised time functions and finiteness of
  the {L}orentzian distance}, J. Geom. Phys. \textbf{106} (2016), 108--121.

\bibitem[Str06]{Str06}
A.~Strohmaier, \emph{On noncommutative and pseudo-{R}iemannian geometry}, J.
  Geom. Phys. \textbf{56} (2006), no.~2, 175--195.

\bibitem[Sui04]{Sui04}
W.~D. {\noopsort{Suijlekom}}van~Suijlekom, \emph{The noncommutative
  {L}orentzian cylinder as an isospectral deformation}, J. Math. Phys.
  \textbf{45} (2004), 537--556.

\bibitem[Tra92]{Tra92}
A.~Trautman, \emph{Spinors and the {D}irac operator on hypersurfaces. {I}.
  {G}eneral theory}, J. Math. Phys. \textbf{33} (1992), no.~12, 4011--4019.

\end{thebibliography}

\providecommand{\noopsort}[1]{}\providecommand{\vannoopsort}[1]{}
\providecommand{\bysame}{\leavevmode\hbox to3em{\hrulefill}\thinspace}
\providecommand{\MR}{\relax\ifhmode\unskip\space\fi MR }
\providecommand{\MRhref}[2]{%
  \href{http://www.ams.org/mathscinet-getitem?mr=#1}{#2}
}
\providecommand{\href}[2]{#2}

\end{document}